\documentclass[11pt]{article}
\usepackage[utf8]{inputenc}

\usepackage{amsmath, amssymb, amsthm}
\usepackage{mathtools}

\usepackage{xcolor}
\definecolor{ForestGreen}{rgb}{0.1333,0.5451,0.1333}
\usepackage[colorlinks=true,linkcolor=ForestGreen,citecolor=blue]{hyperref}

\usepackage[margin=1in]{geometry}
\usepackage{graphics}
\usepackage{pifont}
\usepackage{tikz}
\usepackage{bbm}
\usepackage[T1]{fontenc}
\DeclareMathOperator*{\argmax}{argmax}
\DeclareMathOperator*{\argmin}{argmin}
\usetikzlibrary{arrows.meta}
\usepackage{environ}
\usepackage{framed}
\usepackage{url}
\usepackage{algorithm}
\usepackage[algo2e,linesnumbered,lined,ruled,boxed]{algorithm2e}
\usepackage{algpseudocode}
\usepackage[noabbrev,capitalize,nameinlink]{cleveref}
\crefname{equation}{}{}
\usepackage[labelfont=bf]{caption}
\usepackage{cite}
\usepackage{framed}
\usepackage[framemethod=tikz]{mdframed}
\usepackage{appendix}
\usepackage{graphicx}
\usepackage[textsize=tiny]{todonotes}
\usepackage{tcolorbox}
\allowdisplaybreaks[1]

\newcommand\remove[1]{}

\newtheorem{lemma}{Lemma}[section]
\newtheorem{theorem}{Theorem}
\newtheorem*{lemma*}{Lemma}

\newtheorem*{corollary*}{Corollary}

\theoremstyle{definition}

\newtheorem*{theorem*}{Theorem}
\newtheorem{definition}[lemma]{Definition}

\newtheorem*{rem*}{Remark}

\newcommand{\wh}{\widehat}

\newcommand\R{\mathbb{R}}

\newcommand\Z{\mathbb{Z}}

\newcommand\E{\mathbb{E}}

\newcommand{\eps}{\varepsilon}
\renewcommand{\O}{\widetilde{O}}

\renewcommand{\l}{\langle}
\renewcommand{\r}{\rangle}

\newcommand{\assign}{\leftarrow}
\newcommand{\otilde}{\O}
\renewcommand{\forall}{\mathrm{\text{ for all }}}

\newcommand{\ma}{\mathbf{A}}

\newcommand{\mDiag}{\mathrm{\mathbf{diag}}}

\newcommand{\nnz}{\mathrm{nnz}}

\renewcommand{\hat}{\widehat}
\newcommand{\Subsample}{\mathtt{Subsample}}

\foreach \x in {A,...,Z}{%
	\expandafter\xdef\csname m\x\endcsname{\noexpand\mathbf{\x}}
}

\newif\ifrandom
\randomtrue
\newcommand{\defeq}{\stackrel{\mathrm{\scriptscriptstyle def}}{=}}
\newcommand{\norm}[1]{\lVert#1\rVert}

\newcommand{\poly}{{\mathrm{poly}}}	
\newcommand{\one}{\mathbbm{1}}
\newcommand{\rows}{\mathcal{I}}

\newcommand{\indicVec}[1]{\vec{1}_{#1}}

\newcommand{\todolater}[1]{}

\newcommand{\wt}{\widetilde}

\renewcommand{\bar}{\overline}
\newcommand{\hull}{\mathrm{conv}}

\newcommand{\rank}{\mathrm{rank}}
\newcommand{\vzero}{\vec{0}}

\newcommand{\yang}[1]{\textbf{\color{red}[Yang: #1]}}
\newcommand{\arun}[1]{\textbf{\color{blue}[Arun: #1]}}
\newcommand{\sidford}[1]{\textbf{\color{green}[Aaron: #1]}}
\renewcommand{\S}{\mathcal{S}}

\newcommand{\G}{\mathcal{G}}
\newcommand{\A}{\mathcal{A}}
\renewcommand{\H}{\mathcal{H}}

\newcommand{\B}{\mathcal{B}}

\newcommand{\algGroupLever}{\mathtt{GroupLeverageOverestimate}}

\crefname{algocf}{Algorithm}{Algorithms}

\author{
Arun Jambulapati \\
University of Washington	\\
\texttt{jmblpati@uw.edu}
	 \and 
Yang P. Liu \\
Stanford University \\
\texttt{yangpliu@stanford.edu}
\and 
Aaron Sidford \\
Stanford University \\
\texttt{sidford@stanford.edu}
}

\begin{document}

\title{Chaining, Group Leverage Score Overestimates,\\ and Fast Spectral Hypergraph Sparsification}

\clearpage\maketitle

\begin{abstract}
We present an algorithm that given any $n$-vertex, $m$-edge, rank $r$ hypergraph constructs a spectral sparsifier with $O(n \varepsilon^{-2} \log n \log r)$ hyperedges in nearly-linear $\widetilde{O}(mr)$ \footnote{Throughout the paper, $\widetilde{O}(\cdot)$ hides $\mathrm{poly}\log(m, n, r)$ factors} time. This improves in both size and efficiency over a line of work \cite{BST19,KKTY21b,KKTY21} for which the previous best size was $O(\min\{n \varepsilon^{-4} \log^3 n,nr^3 \varepsilon^{-2} \log n\})$ and runtime was $\widetilde{O}(mr + n^{O(1)})$. 

\paragraph{Independent Result:} In an independent work, Lee \cite{Lee22} also shows how to compute a spectral hypergraph sparsifier with $O(n \varepsilon^{-2} \log n \log r)$ hyperedges. 

\end{abstract}
 
\section{Introduction}
\label{sec:intro}

The problem of \emph{sparsification} asks to reduce the size of an object while preserving some desired properties. For example, a \emph{cut sparsifier} reduces the number of edges in a graph while approximately preserving the total weight of each cut, and a \emph{spectral sparsifier} reduces the number of edges in a graph while approximately preserving the spectral form of the Laplacian, or equivalently the electrical energy of any potentials. Over the last few decades, a variety of efficient and effective algorithms have been developed for these notions of graph sparsification \cite{BK15,BSS14,SS11}.

In recent years there has been a variety of work seeking to sparsify more complex objectives (see e.g. \cite{MMWY21}). One such example is the problem of spectral hypergraph sparsification (see \cite{SY19} for discussion), which has seen significant attention. In this setting, formalized by \cite{SY19}, we have a hypergraph $\G = (V, E, v)$, where $V$ denotes a finite vertex set, $E$ denotes the edge set, and $v \in \R_{\ge0}^E$ denotes non-negative edge weights. Here the edge set is a collection of subsets of $V$ of size at least two, i.e.\ $E \subseteq \{0,1\}^V$ and $|S| \ge 2$ for all $S \in E$ and $G$ is said to be of rank $r$ if the cardinality of each hyperedge is at most $r$, i.e. $|S| \leq r$ for all $S \in E$. Consequently, when $r = 2$ a hypergraph is simply an undirected graph. For every vector $x \in \R^V$ we define its associated \emph{energy} in $G$ as
\begin{align} f_{\G}(x) \defeq \sum_{S \in E} v_S \max_{i, j \in S} (x_i - x_j)^2. \label{eq:normalenergy} \end{align}
The problem of spectral hypergraph sparsification asks to produce a hypergraph $\H$ consisting of a small subset of the hyperedges of $\G$, possibly reweighted, whose energy approximates the energy of $\G$ on all vectors $x \in \R^n$ up to a $(1+\eps)$ multiplicative approximation.

When $r = 2$, spectral hypergraph sparsification exactly reduces to spectral sparsification, where it is known that a random-sampling algorithm can produce a sparsifier with $O(n\eps^{-2} \log n)$ edges \cite{SS11} (it is known how to improve this bound to $O(n\eps^{-2})$ with more adaptive edge choices \cite{BSS14}). For spectral hypergraph sparsification, a line of work \cite{BST19,KKTY21b,KKTY21} has shown that every hypergraph $\G$ admits a sparsifier with a nearly-linear $O(n\eps^{-4} \log^3 n)$ edges, and is surprisingly independent of the rank $r$. Additionally, \cite{KKTY21} proved that there is a random-sampling algorithm that constructs such a sparsifier with high probability in time $\O(mr + n^{O(1)})$.

Building on this line of work, in particular \cite{KKTY21}, the main result of this paper is the following \Cref{thm:hypergraph}.
\begin{theorem}[Hypergraph Sparsification]
\label{thm:hypergraph}
There is an algorithm that given a rank $r$ hypergraph $\G = (V, E, v)$ with $n$ vertices computes a $(1+\eps)$-approximate spectral hypergraph sparsifier with $O(n\eps^{-2} \log n \log r)$ hyperedges in nearly-linear time, i.e. $\O(\sum_{S \in E} |S|)$, with high probability in $n$.
\end{theorem}

This result consists of two key ingredients. First, we introduce a broad class of sampling probabilities which we call \emph{group leverage score overestimates} (\cref{def:group_over}).
While the sampling weights in \cite{KKTY21} took time $\O(mr + n^{O(1)})$ to calculate, we show how to compute our more general weights in nearly-linear, $\O(mr)$, time. Second, we use the generic chaining machinery developed by Talagrand \cite{book} to show that that the sampling algorithm of \cite{KKTY21}, with group leverage score overestimates, actually produces a $(1+\eps)$-spectral hypergraph sparsifier with $O(n\eps^{-2} \log n \log r)$ edges. This improves over the previous bounds of $O(n\eps^{-4} \log^3 n)$ \cite{KKTY21}, and $O(n\eps^{-2}r^3 \log n)$ \cite{BST19}.

\paragraph{Paper Organization.} In the remainder of the introduction we discuss our high level setup required to show \cref{thm:hypergraph}. After providing notation in  \cref{sec:notation}, in \cref{sec:intromatrix} we describe a more general matrix formulation of hypergraph sparsification that we work with, which we call a \emph{matrix hypergraph}. In \cref{sec:introweights} we then introduce our new definition of group leverage score overestimates (\cref{def:group_over}) which can be computed efficiently and still suffices for sampling when constructing sparsifiers. Then, in \cref{sec:introchaining}, we provide a high-level overview of the ideas behind generic chaining, which we use to improve the size bound to $O(n\eps^{-2}\log n \log r)$.

After the introduction, we provide our efficient algorithm for computing group leverage score overestimates in \cref{sec:lever_compute}. In \cref{sec:dudley}, we analyze a sampling algorithm that produces a spectral hypergraph sparsifier by using a simplified form of chaining known as Dudley's inequality. The number of hyperedges will be $O(n\eps^{-2} \log^3 n)$. In \cref{sec:chaining} we use the powerful generic chaining machinery presented in \cite{book}, specifically the growth functional framework, to improve the hyperedge bound to $O(n\eps^{-2} \log n \log r)$.

\subsection{General Notation}
\label{sec:notation}
Throughout, we use $C$ (or $C$ with a subscript denoting a lemma, theorem, or equation number for clarity) to denote a universal constant. We let $\Z_{\ge\alpha} = \Z \cap [\alpha, +\infty)$ and $\R_{\ge\alpha} = \R \cap [\alpha, +\infty)$. We define $\indicVec{i}$ to be indicator vectors for coordinate $i$, and let $\vec{0}$ be the zero vector. We let $\nnz(\mA)$ denote the number of nonzero entries in a matrix $\mA$. We use $\dagger$ to denote the Moore-Penrose pseudoinverse of a matrix.. We assume all $\log$s are base $e$ unless otherwise denoted. We say that an algorithm succeeds \emph{with high probability in $n$} if for any constant $C \ge 1$, there is some choice of constants in the algorithm that makes it have success probability at least $1 - n^{-C}$. The reader should think of $C$ as fixed but arbitrary throughout the paper. The constants in our main result \cref{thm:hypergraph} will depend on this constant $C$ (see \cref{thm:chaining}).

\subsection{A Matrix Generalization of Hypergraph Sparsification}
\label{sec:intromatrix}
We introduce a generalization of hypergraph sparsification to general matrices that we use throughout the paper.
Let $a_1, \dots, a_m \in \R^n$ denote the rows of a matrix $\mA \in \R^{m \times n}$, let $\S = \{S_1, \dots, S_k\}$ be a partition of $[m]$ into $k$ subsets, so $k = |\S|$, and let each set $S_i$ have a non-negative weight $v_i$, forming a vector $v \in \R^k$. We denote the tuple of the matrix $\mA$, the partition $\S$, and the weights $v$ as the (matrix) hypergraph $\G = (\S, \mA, v)$ (henceforth referred to simply as a \emph{hypergraph}). We define the \emph{rank} of a matrix hypergraph as $r = \max_{S \in \S} |S|$. We will assume $r \ge 2$ throughout, as we can duplicate rows $a_i$. We let $f_{\G}: \R^d \to \R$ denote the \emph{energy function of $\G$} where $f_{\G}(x)$, the energy of $\G$ of $x$, is defined as 
\begin{align}
    f_{\G}(x) \defeq \sum_{i\in[k]} v_i \max_{j \in S_i} \l a_j, x\r^2. \label{eq:energy}
\end{align}
Note that \eqref{eq:energy} generalizes the hypergraph energy in \eqref{eq:normalenergy}, because for a hypergraph with $n = |V|$ vertices and $k = |E|$ hyperedges, a hyperedge of weight $v$ containing the vertices $F \subseteq V$ can be captured with the vectors $a_i = (\indicVec{u_1} - \indicVec{u_2})$ for all pairs $u_1, u_2 \in F$ with weight $v$. The rank of the matrix hypergraph will be at most $r(r-1)/2$ if the hypergraph has rank $r$. By definition, in this case the matrix $\mA$ will be the incidence matrix of some multigraph $G$. We will call matrix hypergraphs where $\mA$ comes from a normal hypergraph spectral sparsification instance \emph{graphical hypergraphs}. We will use the term graphical hypergraphs primarily in \cref{thm:overestimation_graph_hypergraph}, when we show how to efficiently compute sampling weights for them. 

We show that this matrix generalization of hypergraph energy can be sparsified essentially as well as graphical hypergraphs. Here, we say that a matrix hypergraph $\H$ is a $(1+\eps)$-approximate spectral sparsifier of $\G$ if $(1+\eps)^{-1}f_{\H}(x) \le f_{\G}(x) \le (1+\eps)f_{\H}(x)$ for all $x \in \R^n$.
\begin{theorem}[Matrix hypergraph sparsification]
\label{thm:matrix}
There is an algorithm that given a matrix hypergraph $\G = (\S, \mA, v)$ with $\mA \in \R^{m \times n}$, and $r = \max_{S \in \S} |S|$ computes a $(1+\eps)$-approximate spectral hypergraph sparsifier with $O(n\eps^{-2}\log m \log r)$ hyperedges in $\O(\nnz(\mA) + n^\omega)$ time.
\end{theorem}

\paragraph{Unit matrix hypergraphs:} A nice benefit of the general matrix setup is that we may assume that the base hypergraph $\G = (\S, \mA, v)$ has unit weights, i.e. all $v_i = 1$. This is without loss of generality, by scaling rows of $\mA$, i.e.\ $\mA \gets \mV^{1/2} \mA$ for $\mV = \mDiag(v)$. We make this assumption for the remainder of the paper, and denote unit matrix hypergraphs as $\G = (\S, \mA)$, omitting the $v$.

\subsection{Group Leverage Score Overestimates}
\label{sec:introweights}
\label{sec:group_lever_properties}

A critical component of the $O(n\eps^{-4}\log^3 n)$ size sparsifier in the previous work was the \emph{balanced weight assignment} \cite[Definition 5.1]{KKTY21} (elaborated on after \cref{def:over}) This was used to prove that the sum of ``importances'' of the hyperedges was bounded by at most $n$, generalizing the notion of leverage scores in graphs. In this paper, we introduce a weaker version of a balanced weight assignment, in that we only enforce a one-sided inequality and a total size bound, instead of the substantially tighter condition in \cite{KKTY21}.
\begin{definition}[Group Leverage Score Overestimates]
\label{def:over}\label{def:group_over}
We say that $\tau \in \R^\S_{\geq 0}$ are \emph{$\nu$-(bounded group leverage score) overestimates} for a unit hypergraph $\G = (\S, \mA)$ with $\ma \in \R^{m \times n}$ if $\norm{\tau}_1 \leq \nu$ and there exist an associated set of weights, $w \in \R_{\ge0}^m$, such that $\sum_{j \in S_i} w_j = 1$ for all $i \in [k]$, and $\max_{j\in S_i} a_j^\top (\mA^\top \mW \mA)^\dagger a_j \le \tau_i$ for all $i\in[k]$ where $\mW = \mDiag(w)$.
\end{definition}
Our goal is to give an algorithm which computes group leverage score overestimates $\nu$ with $\sum_{i\in[k]} \nu_i = O(n)$.
Compared to our \cref{def:group_over}, the balanced weight assignment in \cite[Definition 5.1]{KKTY21} enforced that for all $j \in S_i$, either $w_j = 0$ or $a_j^\top(\mA^\top \mW \mA)^{\dagger} a_j \in [\tau_i/\gamma, \tau_i]$ for a constant $\gamma = O(1)$, without initially enforcing that $\sum_{i\in[k]} \tau_i \le O(n)$. However, it is not difficult to show that this stronger condition implies that $\sum_{i\in[k]} \tau_i \le \gamma n$ (see \cite[Lemma 6.1]{KKTY21}). One reason the balanced weight assignment is a natural definition is that when $\gamma = 1$, the weights $w \in \R^m$ producing the assignment are a minimizer of the convex optimization problem
\[ \min_{\substack{w \in \R_{\ge0}^m \\ \sum_{j\in S_i} w_j = 1 \forall i \in [k]}} -\log\det(\mA^\top \mW \mA). \]
This is essentially the \emph{spanning tree potential} in \cite{KKTY21}, by the matrix tree theorem.

Nonetheless, we show that the weaker notion in \cref{def:over} still suffices to sampling, as long as $\sum_{i\in[k]} \tau_i \le O(n)$. Precisely, we analyze the following simple sampling algorithm (variants of which were studied in \cite{SS11,KKTY21}) where an edge $e$ is kept with probability $p_i \defeq \min\{1, \rho \cdot \tau_i\}$ for an oversampling parameter $\rho$ (generally $\poly(\log n, \eps^{-1})$), and upweighted by a factor of $p_i^{-1}$ so that its value is the same in expectation.
\begin{algorithm2e}[!ht]
\caption{$\Subsample(\G = (\S, \mA), \tau \in \R_{\ge0}^k, \rho)$}
\label{alg:subsample}
\SetKwInOut{Input}{input}
\Input{Rank $r$ unit hypergraph $\G = (\S, \mA)$, group leverage score overestimates $\tau$ (\cref{def:group_over}), and oversampling parameter $\rho$}
Initialize a vector $v \in \R^k$. \\
\For{$i \in [k]$}{
    $p_i \assign \min\{1, \rho \cdot \tau_i\}$. \\
    Set $v'_i \assign p_i^{-1}$ with probability $p_i$, and $0$ otherwise.
}
Return $\H \defeq (\S, \mA, v')$. \tcp{Can remove all sets $S_i$ of $\S$ in $\H$ where $v'_i = 0$}
\end{algorithm2e}

To understand why group leverage scores are useful for subsampling, we introduce the following facts which ultimately show that group leverage score overestimates upper bound the maximum contribution of each coordinate $i \in [k]$ to the total energy.
\begin{lemma}
\label{lemma:xawax}
For any unit hypergraph $\G = (\S, \mA)$ with $\ma \in \R^{m \times n}$ and $w \in \R_{\ge0}^m$ where $\sum_{j \in S_i} w_j = 1$ for all $i\in[k]$, $x^\top \mA^\top \mW \mA x \le f_{\G}(x)$ for all $x \in \R^n$.
\end{lemma}
\begin{proof}
Note that $\sum_{j\in S_i} w_j \l a_j, x\r^2 \leq \max_{j \in S_i} \l a_j, x\r^2$ for all $i\in [k]$ since $\sum_{j \in S_i} w_j = 1$. Hence
\[
x^\top \mA^\top \mW \mA x 
    = \sum_{i\in[k]} \sum_{j\in S_i} w_j \l a_j, x\r^2 
    \le \sum_{i\in[k]} \max_{j \in S_i} \l a_j, x\r^2
    \le f_{\G}(x). \qedhere
\]
\end{proof}
\begin{lemma}
\label{lemma:over} For any group leverage scores $\tau \in \R^{\S}_{\geq 0}$ and associated weights $w \in \R^{m}_{\geq 0}$ for unit hypergraph $\G = (\S, \mA)$ with $\ma \in \R^{m \times n}$, $\max_{j \in S_i} \l a_j, x \r^2 \le \tau_i \cdot x^\top \mA^\top \mW \mA x$ for all $i \in [k]$.
\end{lemma}
\begin{proof}
We can assume that $x^\top \mA^\top \mW \mA x = 1$ by scaling. Note that
\[ \max_{x^\top \mA^\top \mW \mA x = 1} \l a_j, x\r^2 = a_j^\top(\mA^\top \mW \mA)^\dagger a_j \le \tau_i \] for all $j \in S_i$ by \cref{def:over}.
\end{proof}
Combining \cref{lemma:xawax,lemma:over} shows that $\max_{j \in S_i} \l a_j, x\r^2 \le \tau_i f_{\G}(x)$, for all $x$, i.e. coordinate $i \in [k]$ can only contribute $\tau_i$ fraction of the hypergraph energy. Intuitively, this means that sampling proportional to $\tau_i$ should produce a sparsifier, though formalizing this intuition and achieving tight bounds is challenging. This is the main goal of \cref{sec:dudley,sec:chaining}.

It is worth remarking on some general connections between the group leverage scores defined in \cref{def:over}, and similar notions defined for Lewis weights. In general, there are several settings where iterative/contractive procedures produce weights satisfying a one-sided bound, and where such a bound suffices for applications. Our iterative algorithm for computing group leverage score overestimates (\cref{alg:lever_overestimate}) is inspired by the algorithm of \cite{CCLY19} for computing an approximate John ellipse, corresponding to $\ell_\infty$ \emph{Lewis weights} \cite{CP15}. The notion of approximate weights in \cite{CCLY19} is very similar to \cref{def:over}. Additionally, a one-sided $\ell_p$ Lewis weight computation sufficed for the algorithm of \cite{JLS22} for $\ell_p$ regression.

\subsection{Overview of Chaining}
\label{sec:introchaining}
In the section we introduce the basic intuition behind chaining methods, in particular when applied to analyze our sparsification algorithm which samples by group leverage score overestimates (\cref{alg:subsample}). The sampling algorithm proposed in \cref{alg:subsample} keeps a hyperedge $S_i \in \S$ and assigns it weight $p_i^{-1}$ for some probability $p_i$ to produce a hypergraph $\H$. We want to prove that the value of $f_{\G}(x)$ is preserved up to a multiplicative $(1+\eps)$ approximation \emph{for all} $x \in \R^n$. Even though it is straightforward to show that $f_{\G}(x)$ is preserved up to $(1+\eps)$-multiplicatively for each \emph{fixed} $x \in \R^n$, there are infinitely many $x \in \R^n$ which prevents us from applying a union bound. Even a na\"{i}ve discretization leaves exponentially many $x$ to check.

The idea behind chaining is to introduce a sequence of finer and finer $\epsilon$-nets to approximate each $x$ at different scales. Define $B$ as the unit ball of $f_{\G}$, i.e. $B = \{x : f_{\G}(x) \le 1\}$. Consider finite subsets $T_0, T_1, \dots \subseteq B$ of increasing size, which are our nets. For each $N \ge 0$ let $x_N \in T_N$ be the closest point to $x$ in the metric $d(\cdot,\cdot)$ which we define shortly. Write
\[ f_{\G}(x) = f_{\G}(x_0) + \sum_{N \ge 0} f_{\G}(x_{N+1}) - f_{\G}(x_N), \] where the sum converges because $x_N \to x$. Let $\H$ be the subsampled hypergraph, so we get
\begin{align} |f_{\G}(x) - f_{\H}(x)| \le |f_{\G}(x_0) - f_{\H}(x_0)| + \sum_{N \ge 0} |(f_{\G}(x_{N+1}) - f_{\G}(x_N)) - (f_{\H}(x_{N+1}) - f_{\H}(x_N))| \label{eq:introchain} \end{align} by the triangle inequality.
Thus we want to bound $|(f_{\G}(y) - f_{\G}(z)) - (f_{\H}(y) - f_{\H}(z))|$ for several pairs $(y, z)$. To analyze this, note that $\E_{\H}[(f_{\G}(y) - f_{\G}(z)) - (f_{\H}(y) - f_{\H}(z))] = 0$ by the definition of $\H$. If we define a \emph{distance}
\[ d(y, z) := \mathrm{Var}_{\H}[f_{\H}(y) - f_{\H}(z)]^{1/2} = \E_{\H}[((f_{\G}(y) - f_{\G}(z)) - (f_{\H}(y) - f_{\H}(z)))^2]^{1/2}, \] by Hoeffding's inequality we know that
\[ \Pr[|(f_{\G}(y) - f_{\G}(z)) - (f_{\H}(y) - f_{\H}(z))| \ge \kappa d(y, z)] \le 2\exp(-2\kappa^2). \]
Hence, the probability that for $N \ge 0$, parameter $\kappa_N$, and all $x_{N+1} \in T_{N+1}, x_N \in T_N$, \begin{align} |(f_{\G}(x_{N+1}) - f_{\G}(x_N)) - (f_{\H}(x_{N+1}) - f_{\H}(x_N))| \le \kappa_N d(x_N, x_{N+1}) \label{eq:kappan} \end{align} is at least $1 - 2|T_N||T_{N+1}| \exp(-\kappa_N^2)$. At this point, up to constants, it makes sense to set $|T_N| = 2^{2^N}$ for all $N$, and $\kappa_N = C \cdot 2^{N/2}$ for sufficiently large constant $C$, so that $2|T_N||T_{N+1}| \exp(-\kappa_N^2) \le \exp(-2^{2^N})$. Thus \eqref{eq:kappan} holds for all $N \ge 0$ by a union bound. Plugging this all back into \eqref{eq:introchain} and using that $d(\cdot,\cdot)$ satisfies the triangle inequality (at least up to constants), proves the main chaining theorem, which we formally state in \cref{thm:chaining}.

With the chaining theorem in hand, proving the desired sampling bounds in \cref{thm:hypergraph} reduces to constructing sets $T_N$ such that the distances $d(x, T_N) = \min_{y \in T_N} d(x, y)$ are suitably bounded. Surprisingly, the celebrated majorizing measures theorem \cite{Fer75,Tal87} says in variants of the above setting when the sampling distribution is Gaussian instead of Bernoulli (as in our case), this proof method is optimal, i.e. there exist nets $T_N$ with $|T_N| = 2^{2^N}$ that achieve the true optimal bound. We also believe that in our hypergraph sparsification setting, the Gaussian and Bernoulli sampling processes behave similarly. However, the majorizing measures theorem does not shed light on how to construct the sets $T_N$. Many previous works on chaining have thus settled for suboptimal bounds such as Dudley's inequality \cite{Dud67}, which we use in \cref{sec:dudley} to achieve an $O(n\eps^{-2}\log^3 n)$ bound, or rely on analysis frameworks which require additional structure. Towards achieving a better bound in \cref{sec:chaining}, we apply a powerful \emph{growth function} framework of Talagrand which shows how to construct the sets $T_N$ given access to a family of functions satisfying a certain growth condition (\cref{def:growth}). We defer a more detailed explanation of our application of the growth function framework and deviations from prior work (in particular, the proof of matrix Chernoff for rank one matrices \cite{R96,book}) to the start of \cref{sec:chaining}.

\subsection{Related Work}
\label{sec:related}
We discuss relevant related work on chaining and sparsification by sampling.
\paragraph{Hypergraph spectral sparsification.} Previous works showed that hypergraphs admit sparsifiers with $O(n^3\eps^{-2})$ \cite{SY19}, $O(n\eps^{-2}r^3\log n)$ \cite{BST19}, $O(nr(\eps^{-1}\log n)^{O(1)})$ \cite{KKTY21b}, and finally $O(n\eps^{-4}\log^3 n)$ \cite{KKTY21} hyperedges. The independent and concurrent work of Lee \cite{Lee22} also used chaining to show that hypergraphs admit spectral sparsifiers with $O(n\eps^{-2} \log n \log r)$ hyperedges, matching our \cref{thm:hypergraph}. The result \cite{BST19} also used chaining methods, however, their chaining was over the space of matrices, instead of vectors as is done in this paper.

\paragraph{Hypergraph cut sparsification.} The problem of hypergraph cut sparsification  \cite{KK15,CX18} asks to maintain the energy of the hypergraph (see \eqref{eq:normalenergy}), but only for vectors $x \in \{0, 1\}^n$. This generalizes the notion of cut sparsification in graphs. In this setting it is known how to construct hypergraph cut sparsifiers with $O(n\eps^{-2}\log n)$ edges with a random sampling algorithm based on a different notion of ``balanced weight assignments'' \cite{CKN20}. Their algorithm runs in time $\O(mr + n^{O(1)})$. Because hypergraph spectral sparsification strictly generalizes cut sparsification, our \cref{thm:hypergraph} produces a hypergraph cut sparsifier in runtime $\O(mr)$, albeit with $O(n\eps^{-2} \log n \log r)$ hyperedges instead of $O(n\eps^{-2}\log n)$ as shown in \cite{CKN20}.

\paragraph{Other sparsification objectives.} In general, one can study sparsification of functions $f: \R^n \to \R_{\ge0}$ defined as $f(x) = \sum_{i\in[k]} f_i(x)$. When $f_i(x) = \l a_i, x\r^2$ for a vector $a_i \in \R^n$, this is exactly spectral sparsification of matrices and is now well-understood using tools such as the matrix Chernoff bound. On the contrary, for other functions $f_i(x)$, the best known sparsification results often proceed via chaining methods. Nearly tight (up to logarithmic factors) sparsification results are known for sparsification of $\ell_p$ norms of matrices, i.e. $f(x) = \|\mA x\|_p^p = \sum_{i \in [k]} |\l a_i, x\r|^p$ for all $p \in [0, \infty)$ \cite{BLM89,Tal90,Tal95,SZ01,Sch11}, and the proofs generally rely on combining chaining methods with $\ell_p$ Lewis weights, a natural importance measure for rows of $\mA$ analogous to our group leverage scores (\cref{def:group_over}). For more discussion on $\ell_p$ norm sparsification, see \cite{CP15}.

Sparsification of several additional convex functions, including Tukey and Huber losses, gamma functions for $\ell_p$ regression, Orlicz norms, etc., is studied in \cite{MMWY21}. The analysis uses chaining methods, among other techniques. 

\paragraph{Future work.} The authors are optimistic that the methods in \cite{KKTY21}, this paper, and \cite{Tal95}, can provide sparsification results for ``$\ell_p$ hypergraph sparsification'' for $p \in [1, 2]$, i.e.\ when the energy function is $f_{\G}(x) \defeq \sum_{i\in[k]} \max_{j\in S_i} |\l a_j, x\r|^p$, or even beyond. This paper leaves these questions as an interesting direction for future work.

\section{Group Leverage Score Overestimates}
\label{sec:lever_compute}

In this section we provide and analyze efficient algorithms for computing group leverage score overestimates as defined in \Cref{def:over}. Our principal subroutine is the following \Cref{alg:lever_overestimate} which turns an algorithm for computing leverage score overestimates for row-reweightings of a matrix $\mA$ into \emph{group} leverage score overestimates for a hypergraph induced by $\mA$. In this section we introduce leverage scores, their overestimates, and procedures for computing them, introduce and analyze \Cref{alg:lever_overestimate}, and then use these results to compute group leverages score overestimates for matrix hypergraphs and graphical hypergraphs.

First, we introduce leverage scores (\cref{lemma:levscore}) as well as leverage score overestimates and algorithms for computing them (\Cref{def:lever_overestimate}).
\begin{definition}[Leverage scores]
\label{lemma:levscore}
Given a matrix $\mA \in \R^{m \times n}$, the leverage score of row $i \in [m]$ is defined as $\sigma_i(\mA) \defeq a_i^\top (\mA^\top \mA)^\dagger a_i$. Let $\sigma(\mA) \in \R^m$ be the corresponding vector of leverage scores.
\end{definition}
It is standard that $\sum_{j \in [m]} \sigma_j(\mA) = \mathrm{\rank}(\mA) \le n$. Additionally, $\sigma_i(\mA) \in [0, 1]$ and $\sigma_i(\mA) = 0$ if and only if $a_i = \vzero$. In all hypergraphs in this paper we assume that it is not the case that $a_i = \vzero$ as it would make no contribution to the energy. It is known how to estimate the leverage scores to constant accuracy in $\O(1)$ calls to linear system solvers for $\mA^\top \mD \mA$ for positive diagonal matrices $\mD$ (see \cref{thm:levapprox}).

In the following definition we overload the term ``overestimate'' with \Cref{def:over} when it is clear if the subject is a matrix or a hypergraph. 
\begin{definition}[Leverage Score Overestimates]
\label{def:lever_overestimate}
\label{def:lever_over}
We call $\wt{\sigma} \in \R^m$, \emph{$\nu$-(bounded leverage score) overestimates for $\mA \in \R^{m \times n}$} if $\norm{\wt{\sigma}}_1 \leq \nu$ and $\wt{\sigma} \geq \sigma(\mA)$ entrywise. Further, we call a procedure $\A$ a \emph{$\nu$-(bounded leverage score) overestimator for $\mA$} if on input $w \in \R^{m}_{\geq 0}$ it outputs $\A(w) \in \R^m_{\geq 0}$ which are $\nu$-bounded leverage score overestimates for $\sigma_\mA(w) \defeq \sigma(\mW^{1/2} \mA)$ where $\mW \defeq \mDiag(w)$. 
\end{definition}

Leverage score overestimates have played a prominent role in sparsifcation and linear system solving \cite{ST14,KMP10,KMP11,KOSZ13,PS14,CKMPPRX14,KS16,KLPSS16,JS21}. Our choice of notation in \Cref{def:lever_overestimate} is strongly influenced by these works. Further, there are known efficient algorithms for computing leverage score overestimates in general and faster algorithms in the case of graphs as summarized in the following \Cref{thm:levapprox}.

\begin{theorem}[Leverage score approximation, \cite{SS11,LMP13,CLMMPS15}]
\label{thm:levapprox}\label{thm:lever_over}
There is an algorithm that given a matrix $\mA \in \R^{m\times n}$ produces $O(n)$-overestimates of $\mA$ in $\O(\nnz(\mA) + n^\omega)$ time with high probability in $n$. If $\mA$ is additionally the weighted incidence matrix of a graph, i.e.\ every row $i$ is all zero except for a single $w_i$ and a single $-w_i$ for $w_i \neq 0$, then the runtime improves to $\O(\nnz(\mA))$. 
\end{theorem}

\begin{proof}
In both cases, the cited works compute $\wh{\sigma} \in \R^{m}_{\geq 0}$ with  $\wh{\sigma}_j \in [(1 - \delta) \sigma_j(\mA), (1 + \delta) \sigma_j(\mA)]$ with high probability in $n$ for any $\delta > 0$ in the stated runtimes multiplied by $O(\poly(1/\delta))$. Since $\norm{\sigma(\ma)}_1 = \rank(\mA) \leq n$ the result follows by invoking these algorithms for constant $\delta > 0$ and outputting $(1 - \delta)^{-1} \wh{\sigma}$.
\end{proof}

Given \Cref{thm:levapprox}, it suffices to provide an algorithm which carefully combines $\otilde(1)$ overestimates for matrices to produce overestimates for hypergraphs. We provide an algorithm which does this in \Cref{alg:lever_overestimate}. This algorithm is a natural generalization of the algorithm of \cite{CCLY19} for computing an approximate John ellipse mentioned in \cref{sec:introweights}. The procedure simply iterates on a weight vector $w^{(t)}$, computing $\wt{\sigma}^{(t)}$ as leverage score overestimates for $\sigma_\mA(w)$ (\Cref{line:sigma_approx}), and then letting $w^{(t + 1)}$ be the natural re-normalization of those weights (\Cref{line:w_update}). The procedure then outputs the average of these weights (\Cref{line:w_set}) as the weights associated with an overestimate $\tau \in \R^{k}_{\geq 0}$ where each entry of $\tau$ is an appropriately scaled up aggregation of the computed leverage score overestimates (\Cref{line:tau_set}). 

\begin{algorithm2e}[!ht]
\caption{$\algGroupLever(\G = (\S, \mA), T, \A)$}
\label{alg:lever_overestimate}
\SetKwInOut{Input}{input}
\Input{Rank $r$ unit hypergraph $\G = (\S, \mA)$ with $\ma \in \R^{m \times n}$, iteration count $T \in \Z_{\geq 1}$, and $\nu$-overestimator $\A$ for $\mA$ (\Cref{def:lever_overestimate})}
Initialize $w^{(1)} \in \R^m_{\geq 0}$ with $w^{(1)}_j = 1/|S_i|$ for all $i \in [k]$ and $j \in S_i$ \label{line:initialize}\;
\For{$t = 1$ to $T$}{
    \label{line:sigma_approx}
    $\wt{\sigma}^{(t)} \gets \A(w^{(t)})$
    \tcp*{$\wt{\sigma}^{(t)} \in \R^{m}_{\geq 0}$ with $\norm{\wt{\sigma}^{(t)}}_1 \leq \nu$ and $\wt{\sigma}^{(t)} \geq \sigma_\mA(w^{(t)})$ entrywise}
    Set $w^{(t + 1)} \in \R^{m}_{> 0}$ with 
    $w^{(t + 1)}_j \gets \wt{\sigma}_j^{(t)} / (\sum_{j' \in S_i} \wt{\sigma}_{j'}^{(t)})$ for all $i \in [k]$ and $j \in S_i$ 
    \label{line:w_update}
    \;
}
Set $\tau \in \R^{k}_{\geq 0}$ with $\tau_i \gets \exp(T^{-1} \log r) \cdot \frac{1}{T} \sum_{t \in [T]} \sum_{j \in S_i} \wt{\sigma}_j^{(t)}$ for all $i \in [k]$ \label{line:tau_set} \;
$\bar{w} \gets \frac{1}{T} \sum_{t \in [T]} w^{(t)}$ \label{line:w_set}\;
\textbf{return} $(\tau, \bar{w})$ \;
\end{algorithm2e}

For intuition behind this algorithm, consider the optimal weights $w^*$ and group leverage scores $\tau^*$, corresponding to $\gamma = 1$ as discussed in \cref{sec:introweights}. Precisely, for the hypergraph $\G = (\S, \mA)$ we have that $a_j^\top(\mA^\top \mW^* \mA)^\dagger a_j = \tau_i^*$ for all $i \in [k]$ and $j \in S_i$, unless $w_j = 0$. This can be more compactly written as $[\sigma_{\mA}(w^*)]_j = w_j^* \tau_i^*$ for all $i \in [k]$ and $j \in S_i$. Because $\sum_{j \in S_i} w_j^* = 1$, we know that $\tau_i^* = \sum_{j' \in S_i} [\sigma_{\mA}(w^*)]_{j'}$ and therefore $w_j^* = [\sigma_{\mA}(w^*)]_j/\sum_{j' \in S_i} [\sigma_{\mA}(w^*)]_{j'}$ for all $i \in [k]$ and $j \in S_i$. Thus, \cref{alg:lever_overestimate} can be viewed as simply updating $w^{(t)}$ as if the above equation was an equality, using overestimates for leverage score, and then averaging the weights over all $t \in [T]$.

In \Cref{thm:group_over} we prove that this algorithm does successfully compute leverage score overestimates. In fact, the theorem implies that it suffices to compute $O(n)$-bounded leverage score overestimates of $O(\log r)$ different reweightings of $\mA \in \R^{m \times n}$ in order to compute $O(n)$-bounded group leverage score overestimates of a rank $r$ hypergraph associated with $\mA$. The proof is similar to that of \cite{CCLY19} for computing approximate John ellipses and uses a critical technical tool of it, the convexity of $\log([\sigma_\mA(w)]_j / w_j)$ with respect to $w$ for any $j$.

We note that is not actually clear that $\tau \geq \tau^*$ where $\tau$ are the overestimates produced by \Cref{alg:lever_overestimate} for $\G$ and $\tau^*$ are the optimal group leverage scores discussed earlier. It is an interesting open problem to determine whether or not this is the case and if it is false, the term ``group leverage score overestimates'' is perhaps a misnomer. However, in either case the overestimates produced are sufficient for hypergraph spectral sparsfication as we prove in \cref{sec:dudley,sec:chaining}.

\begin{theorem}[Group Leverage Score Overestimation Algorithm]
\label{thm:over}\label{thm:group_over}\label{thm:group_over_framework}
Given any rank $r$ unit hypergraph $\G = (\S, \mA)$ with $\ma \in \R^{m \times n}$, $T \in \Z_{\geq 1}$, and $\nu$-overestimator $\A$ for $\mA$ (\Cref{def:lever_overestimate}), $\algGroupLever(\G, T, \nu)$ in \Cref{alg:lever_overestimate} outputs $\exp(T^{-1} \log r)\nu$-overestimates $\tau \in \R^\S_{> 0}$ for $\G$ and associated weights $\bar{w} \in \R^{m}_{\geq 0}$. The algorithm can be implemented in $O(mT)$ time plus the time of invoking $\A$ on $T$ different inputs.
\end{theorem}

\begin{proof}
The runtime is immediate from the pseudocode (there are $T$ iterations each of which takes time $O(m)$ plus the time to invoke $\A$) and consequently it suffices to show that $\tau$ are $\exp(T^{-1} \log r)\nu$-overestimates for $\G$ with associated weights $\bar{w} \in \R^{m}_{\geq 0}$. By the definition of $\tau$ (\Cref{line:tau_set}) and $\wt{\sigma}$ (\Cref{line:sigma_approx} and \Cref{def:lever_overestimate}) it follows that
\[
\norm{\tau}_1
= \exp(T^{-1} \log r) \cdot \sum_{i \in [k]} \left[\frac{1}{T} \sum_{t \in [T]} \sum_{j \in S_i} \wt{\sigma}_j^{(t)} \right]
= \frac{\exp(T^{-1} \log r)}{T}
\sum_{t \in [T]} \norm{\wt{\sigma}^{(t)}}_1
\leq \exp(T^{-1} \log r) \nu \, .
\]
Next, for any $i \in [k]$ and $j \in S_i$ since $\log({[\sigma_\mA(w)]_j} / {w_j})$ is convex in $w$ \cite[Lemma 3.4]{CCLY19} it follows that
\begin{align*}
    \log\left(\frac{[\sigma_\mA(\bar{w})]_j}{\bar{w}_j}\right)
    &\le \frac{1}{T} \sum_{t \in [T]} \log \left( \frac{[\sigma_\mA(w^{(t)})]_j}{w^{(t)}_j} \right) 
    \tag{convexity \cite[Lemma 3.4]{CCLY19}}
    \\
    &\le \frac{1}{T} \sum_{t \in [T]} \log \left( \frac{\wt{\sigma}_j^{(t)}}{w^{(t)}_j} \right)
    \tag{Definition of $\wt{\sigma}$ (\Cref{line:sigma_approx} and \Cref{def:lever_overestimate})}\\
    &= \frac{1}{T} \sum_{t \in [T]} \left[\log \left(\frac{w_j^{(t+1)}}{w_j^{(t)}}\right) 
    + \log\left(\sum_{j' \in S_i} \wt{\sigma}_{j'}^{(t)} \right) \right] \\
    &\le \frac{1}{T} \log\left(\frac{w^{(T + 1)}_j}{w^{(1)}_j}\right) + \log\left(\frac{1}{T} \sum_{t \in [T]} \sum_{j' \in S_i} \wt{\sigma}_{j'}^{(t)} \right) \tag{concavity of $\log(\cdot)$} \\
    &= \frac{1}{T} \log\left(\frac{w^{(T + 1)}_j}{w^{(1)}_j}\right) - \frac{1}{T} \log(r) + \log(\tau_i)\,.
    \tag{Definition of $\tau$  (\Cref{line:tau_set})}
\end{align*}
Now observe that $w_j^{(T)} \leq 1$ (since leverage scores are at most $1$)
and $w_j^{(1)} = \frac{1}{|S_j|} \geq \frac{1}{r}$ (by definition of $w_j^{(1)}$ and $r$). Thus $w_j^{(T + 1)} \leq r \cdot {w_j^{(1)}}$ and we have the desired bound as
\[
\tau_i
\ge 
\frac{[\sigma_\mA(\bar{w})]_j}{\bar{w}_j}
= a_j^\top (\mA \bar{\mW} \mA)^\dagger a_j
\text{ where }
\bar{\mW} = \mDiag(\bar{w})
\,. \qedhere
\]
\end{proof}

As an immediate consequence of \Cref{thm:group_over_framework,thm:lever_over} we obtain an efficient algorithm for computing group leverage score overestimates for general hypergraphs.

\begin{theorem}[Efficient Overestimates of General Hypergraphs]
\label{thm:overestimation_general_hypergraph}
There is an algorithm which given any rank $r$ unit hypergraph $\G = (\S, \mA)$ with $\mA \in \R^{m \times n}$ in time $\O(\nnz(\mA) + n^\omega)$ computes $O(n)$-overestimates for $\G$ with high probability in $n$.
\end{theorem}
\begin{proof}
Apply \Cref{thm:group_over_framework} with $T = \Theta(\log r)$ using \Cref{thm:lever_over} to efficiently implement the $O(n)$-overestimator.
\end{proof}

Finally we show how to use \Cref{thm:group_over_framework,thm:lever_over} to obtain an efficient algorithm for computing group leverage score overestimates for graphical hypergraphs. Na\"{i}vely applying these results would yield an algorithm that in $\O(\sum_{i \in [k]} |S_i|^2)$ computes $O(n)$-overestimates for an $n$-node hypergraph with hyperedges $S_1, ..., S_k$. In the following theorem we show how to improve this to $\O(\sum_{i \in [k]} |S_i|)$ using the trick of using stars to overestimate hyperedges \cite[Section 3]{KKTY21}.

\begin{theorem}[Efficient Overestimates of Graphical Hypergraphs]
\label{thm:overestimation_graph_hypergraph}
There is an an algorithm that given any $n$-node graphical hypergraph $\G = (V, E, v)$ in time $\O(\sum_{S_i \in E} |S_i|)$ outputs with high probability in $n$, $O(n)$-overestimates for the matrix unit-hypergraph associated with $\G$.
\end{theorem}
\begin{proof}
Note that the the matrix unit-hypergraph associated with $\G$, is $(\S, \mA)$ where $\mA \in \R^{m \times V}$ where $m = \sum_{S_i \in E} \genfrac{(}{)}{0pt}{1}{|S_i|}{2}$ and each $a,b \in S_i$ with $a \neq b$ has an associated row in $\mA$, which we call $j_{a,b,i}$, that is $\sqrt{v_i} (\indicVec{a} - \indicVec{b})$. Further, each $S_i \in E$ corresponds to a $S_i \in \S$ containing $j_{a,b,i}$ for each $a,b\in S_i$ with $a \neq b$. 

Now, for each $S_i \in E$ fix an arbitrary vertex $a_i \in S_i$. Further, consider the unit hypergraph $(\S',\mA')$ that consists of discarding from $\G$ the rows $j_{a,b,i}$ where it is not the case that $a = a_i$ and $b \neq a_i$. Note that $\mA' \in \R^{m' \times V}$ with $m' = \sum_{S_i \in E} (|S_i| - 1)$ and $\mA'$ is a weighted incidence matrix of a graph. Consequently, using \Cref{thm:group_over_framework} and \Cref{thm:lever_over} we can compute $\tau \in \R^k$ that are $O(n)$-overestimates for $(\S', \mA')$ with associated weights $w' \in \R^{m'}_{\geq 0}$ with high probability in $n$. 

Consequently, to complete the proof it suffices to show that $\tau = 2\tau'$ are $O(n)$-overestimates for $(\S, \A)$. Clearly $\norm{\tau}_1 = 2\norm{\tau'}_1 \leq O(n)$ and consequently it suffices to produce associated weights $w \in \R^{m}_{\geq 0}$. Define such a $w \in \R^{m}_{\geq 0}$ by setting $w_j$ to have the value of the associated entry in $w'_j$ if row $j$ is in both $\mA$ and $\mA'$ and $0$ otherwise. Since $w'$ were the weights associated with $\tau'$, and the only new weights in $w$ are $0$, we clearly have the property that for all $S_i \in \S$ \[ \sum_{j_{i,a,b} : a,b\in S_i \text{ with } a\neq b} w_{j_{i,a,b}} =1. \] The result then follows from the following, where $\mW \defeq \mDiag(w), \mW' \defeq \mDiag(w')$, and $j_{i,a,b} \in S_i$:
\begin{align*}
a_{j_{i,a,b}}^\top (\mA^\top \mW \mA)^\dagger a_{j_{i,a,b}}
&\overset{(i)}{=}
v_i
(\indicVec{a} - \indicVec{b}) ((\mA')^\top \mW' \mA')^\dagger (\indicVec{a} - \indicVec{b})
\\
&\overset{(ii)}{\leq} 
v_i
(\indicVec{a} - \indicVec{a_i}) ((\mA')^\top \mW' \mA')^\dagger (\indicVec{a} - \indicVec{a_i}) + v_i (\indicVec{a_i} - \indicVec{b}) ((\mA')^\top \mW' \mA')^\dagger (\indicVec{a_i} - \indicVec{b}) \\
&\overset{(iii)}{\le} \tau_i' + \tau_i' = \tau
\end{align*}
Here, $(i)$ follows from the definition of $a_{j_{i,a,b}}$ and $w'$, and $(iii)$ follows because $\tau'$ are overestimates for $\mA'$ with weights $w'$. $(ii)$ follows from the triangle inequality for effective resistances in graphs (see \cite{Tetali}). It is worth remarking that if instead set $\tau_i = 4\tau_i'$ (so $\|\tau\|_1 \le O(\|\tau'\|_1)$ still) that we can simply use the triangle inequality for norms in this line.
\end{proof}

\section{Size Bound from Dudley's Inequality}
\label{sec:dudley}
\newcommand{\Cbd}{C_{\hyperref[thm:dudley_bound]{1}}}
\newcommand{\Cchain}{C_{\hyperref[thm:chaining]{2}}}
\newcommand{\Cf}{C_{\hyperref[lemma:mainak17]{4}}}
\newcommand{\Cak}{C_{\hyperref[thm:ak17]{3}}}

In this section, we prove that sampling hyperedges $S_i$ proportional to their overestimates $\tau_i$ in \cref{def:over} produces a sparsifier with high probability. As a warmup to the results in the following \cref{sec:chaining}, we first prove a weaker size bound of $O(n\eps^{-2} \log^3 m)$ using a simple form of chaining. In the context of previous work on chaining, our proof is essentially just applying Dudley's entropy bound \cite{Dud67} instead of the full generic chaining (we elaborate on this after \cref{thm:chaining}). Because the proof is relatively simple and provides nice intuition for the more complicated analysis in \cref{sec:chaining}, we give a self-contained analysis except for an $\ell_\infty$ ball covering theorem from \cite{AK17}.

Specifically, we prove the following theorem.

\begin{theorem}
\label{thm:dudley_bound}
Let $\G = (\S, \mA)$ be a unit hypergraph, and let $\tau$ be given group leverage scores (\cref{def:over}) with valid weights $w$ (which do not need to be known). For any constant $C$, there is an absolute constant $\Cbd$ (depending on $C$) such that $\H = \Subsample(\G, \tau, \rho)$ (\cref{alg:subsample}) with $\rho = \Cbd \eps^{-2}\log^3 m$ satisfies with probability at least $1 - n^{-C}$ that
\[
(1-\eps)f_{\G}(x) \le f_{\H}(x) \le (1+\eps)f_{\H}(x)
\text{ for all }
x \in \R^n\,.
\]
\end{theorem}
If $\tau$ is given by \cref{thm:group_over}, the above yields the claimed edge bound of $O(n\eps^{-2} \log^3 m)$. We will later improve this bound to $O(n\eps^{-2} \log m \log r)$ in \cref{sec:chaining}.

Let us discuss our general proof strategy for \cref{thm:dudley_bound}. In a chaining argument, it is useful to study how the \emph{difference} between energies of two points $x, y \in \R^n$, i.e. $f_{\G}(x) - f_{\G}(y)$, is affected by sampling.
By construction, the sampling is unbiased for any fixed input $x \in \R^n$, so $\E_{\H}[f_{\H}(x)-f_{\H}(y)] = f_{\G}(x)-f_{\G}(y)$. As is standard in chaining setups, we now define a distance function which is an upper bound on the variance.
\begin{definition}[Metric Space]
\label{def:g_and_d}
For a fixed hypergraph $\G = (\S, \mA)$, define $g_i(x) \defeq \max_{j \in S_i} |\l a_j,x\r|$ for all $S_i \in \S$. We let $B$ be the unit ball of the energy function, i.e. $B \defeq \{f_{\G}(x) \le 1 : x \in \R^n\}$. Additionally, for given sampling probabilities $p_i \in (0, 1]$, we define the \emph{distance function} $d : \R^n \times \R^n \rightarrow \R$ for all $x, y \in \R^n$ as
\begin{align}
    d(x, y) \defeq \left(\sum_{i \in [k]} \one_{\{p_i \neq 1\}} p_i^{-1} (g_i(x)^2 - g_i(y)^2)^2 \right)^{1/2}\,.\label{eq:d}
\end{align}
Further, for a finite subset $T \subseteq \R^n$, we define $d(x, T) \defeq \min_{t \in T} d(x, t)$.
\end{definition}
We observe that the functions $g_i$ are convex and satisfy $f_{\G}(x) = \sum_{i\in[k]} g_i(x)^2.$ We formalize additional key properties of the distance function $d$ in the following lemma.
\begin{lemma}\label{lemma:d_properties}
Let $\G = (\S, \mA)$ be a hypergraph, let $p_i \in (0,1]$ be given, and let $d(\cdot, \cdot)$ be defined as in \cref{def:g_and_d}. Let $\H = (\S, \mA, \hat{v})$, where $\hat{v} \in \R^k$ is defined as $\hat{v}_i = p_i^{-1}$ with probability $p_i$ and $0$ otherwise. $d$ satisfies the following properties for any $x,y,z \in \R^n$:
\begin{itemize}
    \item $\mathrm{Var}_{\H}[f_{\H}(x)-f_{\H}(y)] \le d(x, y)^2$
    \item $d(x, z) \le d(x, y) + d(y, z)$.
\end{itemize}
\end{lemma}
\begin{proof}
To bound the variance, note that $f_{\H}(x)-f_{\H}(y)$ is a sum of $k$ independent random variables, where the $i^{th}$ variable is either $0$ or $p_i^{-1} (g_i(x)^2 - g_i(y)^2)$ if $p_i \neq 1$ and always $(g_i(x)^2 - g_i(y)^2)$ otherwise. Thus we have
\[
\mathrm{Var}[f_{\H}(x)-f_{\H}(y)] = \sum_{i \in [k]} \frac{1 - p_i}{p_i} (g_i(x)^2 - g_i(y)^2)^2 \leq d(x,y)^2
\]
as $p_i \in (0,1]$. This shows the first property. For the second property, define $v^x \in \R^k$ as the vector with coordinates $v^x_i \defeq \sqrt{\one_{\{p_i \neq 1\}} p_i^{-1}} g_i(x)^2$ for all $i \in [m]$. Define $v^y, v^z$ similarly. Now the desired bound of the final property follows because $d(x, y) = \|v^x - v^y\|_2$, and by triangle inequality
\[
d(x,z) =
\|v^x - v^z\|_2 \le \|v^x - v^y\|_2 + \|v^y - v^z\|_2
= d(x,y) + d(y,z)
\,. \qedhere \]
\end{proof}

With these facts in hand, we describe our formal chaining setup.
\begin{theorem}[Chaining]
\label{thm:chaining}
Let $\G$, $p_i$, $d$, $\H$ be defined as in \Cref{lemma:d_properties}. For $s \ge \lceil \log_2(\log n) \rceil$, define
\begin{align} \gamma \defeq \inf_{\substack{T_s, T_{s+1}, \dots \\ T_N \subseteq B, |T_N| \le 2^{2^N} \text{ for all } N \ge s}} \sup_{x \in B} 2^{s/2} \cdot d(x, \vec{0}) + \sum_{N \ge s} 2^{N/2} d(x, T_N). \label{eq:gamma} \end{align}
Then there is an absolute constant $\Cchain$ such that with high probability (i.e. at least $1-n^{-C}$, and $\Cchain$ depends on $C$) for all $x \in \R^n$
\[
(1-\Cchain \gamma)f_{\G}(x) \le f_{\H}(x) \le (1+\Cchain \gamma)f_{\H}(x)
\,.
\]
\end{theorem}
To prove \cref{thm:chaining} we first note the following simple application of Hoeffding's inequality.
\begin{lemma}
\label{lemma:easyclaim}
For any subsets $X, Y \subseteq \R^n$ and $K \ge 0$ we have that with probability at least $1 - 2|X||Y| \exp(-2K^2)$ over choices of $\H$ that for all $x \in X, y \in Y$ that $|(f_{\G}(x) - f_{\G}(y)) - (f_{\H}(x) - f_{\H}(y))| \le K \cdot d(x, y)$.
\end{lemma}
\begin{proof}
Note that 
for each pair $x \in X, y \in Y$ we have that $\E[f_{\H}(x)-f_{\H}(y)] = f_{\G}(x)-f_{\G}(y)$ and that $f_{\H}(x)-f_{\H}(y)$ is a sum of $k$ independent random variables, where the $i^{th}$ variable is either $0$ or $p_i^{-1} (g_i(x)^2 - g_i(y)^2)$ if $p_i \neq 1$ and always $(g_i(x)^2 - g_i(y)^2)$ otherwise by \Cref{lemma:d_properties}. Applying Hoeffding's inequality and the definition of $d$ (see \cref{lemma:d_properties}),
\begin{align*} &\Pr_{\H}\left[|(f_{\G}(x)-f_{\G}(y)) - (f_{\H}(x)-f_{\H}(y))| >  K \cdot d(x, y)\right] \\ 
= ~&\Pr_{\H}\left[|\E[f_{\H}(x)-f_{\H}(y)] - (f_{\H}(x)-f_{\H}(y))| > K \cdot d(x, y)\right] \\ 
\le~& 2 \exp\left(- \frac{2K^2 \cdot d(x,y)^2}{\sum_{i \in [k]}  \one_{\{p_i \neq 1\}} p_i^{-1} (g_i(x)^2 - g_i(y)^2)^2 } \right) \\
= ~& 2 \exp\left( - \frac{2K^2 \cdot d(x,y)^2}{d(x,y)^2} \right) = 2 \exp(-2K^2).
\end{align*}
The claim follows by union bounding over all $x \in X, y \in Y$.
\end{proof}
To prove \cref{thm:chaining} we apply \cref{lemma:easyclaim} on all levels $N \ge 0$ and add them up.
\begin{proof}[Proof of \cref{thm:chaining}]
Consider the event $E_N$ that for all $x \in T_N, y \in T_{N+1}$ 
\[ 
|(f_{\G}(x)-f_{\G}(y)) - (f_{\H}(x)-f_{\H}(y))| \le 2\sqrt{C} \cdot 2^{N/2} d(x, y). 
\]
We claim that $\Pr_{\H}[E_N] \ge 1 - 1/2 \cdot n^{-C} 2^{-2^N}$. Indeed this is true by taking $X = T_N, Y = T_{N+1}$, and $K = 2\sqrt{C} \cdot 2^{N/2}$ in \cref{lemma:easyclaim} and noting that
\begin{align*} 2|T_N| |T_{N+1}| \exp(-2K^2) \le 2 \cdot 2^{3 \cdot 2^N} \exp(-8C \cdot 2^N) \le 1/2 \cdot n^{-C} 2^{-2^N}
\end{align*}
because we assume $C \ge 1$ and $2^N \ge 2^s \ge \log n$. Hence all events $E_s, E_{s+1}, \dots$ hold with probability at least $1 - \sum_{N\ge0} 1/2 \cdot n^{-C}2^{-2^N} \ge 1-1/2 \cdot n^{-C}$.

By setting $X = T_s, Y = \{0\}, K = 2\sqrt{C} \cdot 2^{s/2}$ and again applying \cref{lemma:easyclaim} we get that 
\begin{align}
|f_{\G}(x) - f_{\H}(x)| \le 2\sqrt{C} \cdot  2^{s/2}\cdot d(x, \vec{0}) \text{ for all } x \in T_N \label{eq:sevent}
\end{align}
for all $x \in T_s$ with probability at least
\begin{align*}
    1 - |T_s| \exp(-8C \cdot 2^s) \ge 1 - 1/2 \cdot n^{-C}
\end{align*}
because we assume $C \ge 1$ and $2^s \ge \log n$. Hence all events $E_s, E_{s+1}, \dots$ and \eqref{eq:sevent} hold with probability at least $1 - n^{-C}$.
Now, for each $x \in B$ let $x_N = \argmin_{y \in T_N} d(x, y)$. If all events above hold, then
\begin{align*} 
|f_{\G}(x) - f_{\H}(x)| &\le |f_{\G}(x_s) - f_{\H}(x_s)| + \sum_{N \ge s} |f_{\G}(x_N)-f_{\G}(x_{N+1})-(f_{\H}(x_N)-f_{\H}(x_{N+1}))| \\ 
&\le 2\sqrt{C} \cdot 2^{s/2} \cdot d(x_s, 0) + 2\sqrt{C} \sum_{N \ge s} 2^{N/2} d(x_N, x_{N+1}) \\
&\overset{(i)}{\le} 2\sqrt{C} \cdot 2^{s/2} \cdot (d(x_s, x) + d(x, \vec{0})) + 2\sqrt{C} \sum_{N \ge s} 2^{N/2} (d(x, x_N) + d(x, x_{N+1})) \\ 
&\le 2\sqrt{C} \cdot 2^{s/2} \cdot d(x, \vec{0}) + 2\sqrt{C} \sum_{N \ge s} \left(2^{(N-1)/2} + 2^{N/2} \right) d(x, T_N)\\
&\le 4\sqrt{C} \cdot 2^{s/2} \cdot d(x, \vec{0}) + 4\sqrt{C} \sum_{N \ge s} 2^{N/2} d(x, T_N) \leq 4\sqrt{C}\gamma,
\end{align*}
where $(i)$ uses that $d$ is a metric (\cref{lemma:d_properties}). Thus we may set $\Cchain = 4\sqrt{C}$. 
\end{proof}

The goal of the remainder of this section is to bound the quantity in \eqref{eq:gamma} for sampling probabilities $p_i \defeq \min\{1, \rho(m,\eps)\tau_i\}$, where $\rho \defeq \rho(m,\eps)$ is an oversampling parameter. In this section, we will set $\rho = \Cbd \eps^{-2} \log^3 m$: later in \cref{sec:chaining} we modify the chaining argument to show $\rho = C\eps^{-2} \log m \log r$ still suffices for some sufficiently large constant $C$. Because $\sum_{i \in [k]} \tau_i \le O(n)$ by \cref{thm:over}, the hypergraph $\H$ will have $O(\rho n)$ edges with high probability.

We first handle the term $2^{s/2} \cdot d(x, \vec{0})$ in \cref{thm:chaining}. This calculation provides critical intuition for why group leverage scores are sufficient for sampling.
\begin{lemma}[Handling $d(x, \vec{0})$]
\label{lemma:dx0}
For group leverage score overestimates $\tau$ and corresponding weights $w$ (\cref{def:over}), $x \in B$, and $p_i \defeq \min\{1, \rho(m,\eps)\tau_i\}$ we have $d(x, \vec{0}) \le \rho^{-1/2}$.
\end{lemma}
\begin{proof}
Note that $\one_{p_i\neq1}p_i^{-1}\tau_i \le \rho^{-1}$ and $g_i(0) = 0$ for all $i \in [k]$. Hence
\begin{align*}
    d(x, \vec{0}) 
    &= \left(\sum_{i\in[k]} \one_{p_i\neq1}p_i^{-1} g_i(x)^4 \right)^{1/2} 
    \overset{(i)}{\le} \left(\sum_{i\in[k]} \one_{p_i\neq1}p_i^{-1} \tau_i \cdot x^\top \mA^\top\mW\mA x \cdot g_i(x)^2 \right)^{1/2} \\
    &\overset{(ii)}{\le} \rho^{-1/2} (x^\top\mA^\top\mW\mA x)^{1/2} \left(\sum_{i\in[k]} g_i(x)^2 \right)^{1/2} \overset{(iii)}{\le} \rho^{-1/2} f_{\G}(x) \le \rho^{-1/2}.
\end{align*}
Here, $(i)$ follows from \cref{lemma:over}, $(ii)$ follows from $\one_{p_i\neq1}p_i^{-1} \le \rho^{-1}\tau_i^{-1}$ as noted, and $(iii)$ follows from \cref{lemma:xawax}.
\end{proof}
Next we will construct nets $T_N$ for $N \ge C\log m$ for sufficiently large constant $C$ that will show show that the contribution of those terms to \cref{eq:gamma} is negligible. At this scale we have $|T_N| = 2^{2^N} = \exp(\poly(m))$, while there are only $m$ vectors $a_i$. Consequently our net will simply just approximate each inner product $|\l a_i, x\r|$ up to additive $\delta$ accuracy for properly chosen $\delta$. This creates $(1/\delta)^m$ net centers, which is much less than the allowed $2^{2^N} \ge \exp(\poly(m))$.
\begin{lemma}[Large $N$]
\label{lemma:largen}
Consider group leverage score overestimates $\tau$ and corresponding weights $w$ (\cref{def:over}), $x \in B$, and $p_i \defeq \min\{1, \rho(m,\eps)\tau_i\}$. For all $N \ge 0$, there is $T_N \subseteq B$ with $|T_N| \le 2^{2^N}$ and $d(x, T_N) \le 2 \cdot 2^{-2^{N-1}/m}$ for all $x \in B$.
\end{lemma}
\begin{proof}
Fix an $N$. Recall from previous arguments (e.g. \cref{lemma:dx0}) that
\[ \one_{p_i\neq1}p_i^{-1} g_i(x)^2 \le \one_{p_i\neq1}p_i^{-1}\tau_i x^\top\mA^\top\mW\mA x \le \rho^{-1} < 1, 
\] 
by \cref{lemma:over,lemma:xawax}. For each $x \in B$ and $\delta = 2^{-2^N/m}$, consider the vector $v^x \in \R^k$ defined as
\[ v^x_i \defeq \delta \lfloor \one_{p_i\neq1}p_i^{-1} g_i(x)^2 / \delta \rfloor. \]
Note that $v^x_i$ can only be one of at most $(1/\delta)^k \le (1/\delta)^m$ distinct vectors. Thus, we can pick $T_N$ to contain one representative $x \in B$ for each distinct $v^x_i$, and $|T_N| \le 2^{2^N}$ because $(1/\delta)^m = 2^{2^N}$. For any $x \in B$, let $y$ be such that $v^y = v^x$. The result then follows as
\begin{align*}
    d(x, T_N) &\le d(x, y) = \left(\sum_{i \in [k]} \one_{\{p_i \neq 1\}} p_i^{-1} (g_i(x)^2 - g_i(y)^2)^2 \right)^{1/2} \\
    &\le \left(\sum_{i \in [k]} \one_{\{p_i \neq 1\}} p_i^{-1} |g_i(x)^2 - g_i(y)^2| \cdot (g_i(x)^2 + g_i(y)^2) \right)^{1/2} \\
    &\le \left(\sum_{i \in [k]} \delta(g_i(x)^2 + g_i(y)^2) \right)^{1/2} \le \sqrt{\delta(f_{\G}(x) + f_{\G}(y))} \le \sqrt{2\delta} \le 2 \cdot 2^{-2^{N-1}/m}\,. \qedhere
\end{align*}
\end{proof}
This means that the terms $N \ge 4(\log m + \log_2 \log (1/\eps))$ in \eqref{eq:gamma} have low contribution, as \[ \sum_{N \ge 4(\log m + \log_2 \log (1/\eps))} 2^{N/2} \cdot 2 \cdot 2^{-2^{N-1}/m} \le \eps. \]
We conclude this section by bounding the remaining terms in \eqref{eq:gamma}.
\begin{lemma}
\label{lemma:mainak17}
Consider group leverage score overestimates $\tau$ and corresponding weights $w$ (\cref{def:over}), $x \in B$, and $p_i \defeq \min\{1, \rho(m,\eps)\tau_i\}$. For all $N \ge 0$, there is $T_N \subseteq B$ with $|T_N| \le 2^{2^N}$ and $d(x, T_N) \le \Cf \rho^{-1/2}2^{-N/2}\sqrt{\log m}$ where $\Cf$ is an absolute constant.
\end{lemma}
The proof of this lemma uses the following result, which gives a covering of the unit ball with balls of radius $\eta$ in the norm $\max_{i\in[m]} |\l u_i, x\r|$ for unit vectors $u_1, \dots, u_m \in \R^n$.
\begin{theorem}[Theorem VI.1 of \cite{AK17}]
\label{thm:ak17}
Let $u_1, \dots, u_m \in \R^n$ be vectors with $\|u_i\|_2 \le 1$ for all $i \in [m]$, and $\eta > 0$. There is a universal constant $\Cak$ such that the ball $B_2 \defeq \{x : x \in \R^n, \|x\|_2 \le 1\}$ can be covered with at most $S = m^{\Cak/\eta^2}$ subsets $P_1, \dots, P_S$ satisfying
\[ \max_{\substack{i \in [m], j \in [S] \\ x, y \in P_j}} |\l u_i, x-y\r| \le \eta. \]
\end{theorem}
\begin{proof}[Proof of \cref{lemma:mainak17}]
Define $u_j = \tau_i^{-1/2}(\mA^\top\mW\mA)^{-1/2}a_j$ for $j \in S_i$. Note that for any $x \in B$ we have $\|(\mA^\top\mW\mA)^{1/2}x\|_2 \le 1$ by \cref{lemma:xawax} and $x \in B$. In addition, 
\[ 
\l a_j, x\r = \l (\mA^\top\mW\mA)^{-1/2}a_j, (\mA^\top\mW\mA)^{1/2}x \r = \tau_i^{1/2} \l u_j, (\mA^\top\mW\mA)^{1/2}x\r 
\]
and $\|u_j\|_2 = \tau_i^{-1} a_j^\top (\mA^\top\mW\mA)^{-1} a_j \le 1$ by \cref{def:over}. Let $\eta$ satisfy $m^{\Cak/\eta^2} = 2^{2^N}$, so $\eta = \sqrt{\Cak} 2^{-N/2}\sqrt{\log_2 m}$, and let $P_1, \dots, P_S$ be the sets guaranteed by \Cref{thm:ak17} for the vectors $u_i$ and parameter $\eta$. Note that the above facts guarantee 
\begin{align}
\label{eq:thm7}
\max_{\substack{i \in [m], j \in [S] \\ z, w \in P_j}} |\l u_i, z - w\r| \le \eta.
\end{align}
For each $i$ let $v_i$ be an arbitrary point from $P_i$, and let $T_N$ be the set of $(\mA^\top \mW \mA)^{-1/2} v_i$ for all $i$.

For $x \in B$, let $P_j$ be a subset containing $(\ma^\top \mW \ma)^{1/2} x$. Such a $j$ must exist since $P_1, \dots P_s$ cover the unit ball and $\|(\mA^\top\mW\mA)^{1/2}x\|_2 \le 1$. Let $y=(\ma^\top \mW \ma)^{-1/2} v_j$. Then
\begin{align*}
    d(x, T_N) &\le d(x, y) = \left(\sum_{i \in [k]} \one_{\{p_i \neq 1\}} p_i^{-1} (g_i(x)^2 - g_i(y)^2)^2 \right)^{1/2} \\
    &\le \left(\sum_{i \in [k]} \one_{\{p_i \neq 1\}} p_i^{-1} \max_{j \in S_i} (g_i(x) - g_i(y))^2 (g_i(x) + g_i(y))^2 \right)^{1/2} \\
    &\overset{(i)}{\le} \left(\sum_{i \in [k]} \one_{\{p_i \neq 1\}} p_i^{-1} \max_{j \in S_i} \l a_j, x-y\r^2 (g_i(x) + g_i(y))^2 \right)^{1/2} \\
    &\overset{(ii)}{\le} \left(\sum_{i \in [k]} \one_{\{p_i \neq 1\}} p_i^{-1} \tau_i \eta^2 (g_i(x) + g_i(y))^2 \right)^{1/2} \\
    &\overset{(iii)}{\le} \eta \rho^{-1/2} \left(\sum_{i \in [k]}  2(g_i(x)^2 + g_i(y)^2)  \right)^{1/2} \le 2\eta \rho^{-1/2}   =  2\sqrt{\Cak} \rho^{-1/2}2^{-N/2}\sqrt{\log_2 m}.
\end{align*}
Here, $(i)$ follows from $|g_i(x) - g_i(y)| \le g_i(x-y) = \max_{j\in S_i} |\l a_j, x-y\r|$, $(ii)$ holds because $(\mA^\top\mW\mA)^{1/2}x, (\mA^\top\mW\mA)^{1/2}y \in P_j$ and 
\[ 
\l a_j, x-y \r^2 =  \tau_i \l u_j, (\mA^\top\mW\mA)^{1/2}x - (\mA^\top\mW\mA)^{1/2}y\r^2 \leq \tau_i \eta^2
\] by Equation~\Cref{eq:thm7} and $(ii)$ follows from the choice of $p_i$. The claim follows from choosing $\Cf = 2\sqrt{\Cak} (\log 2)^{-1/2}$.  
\end{proof}
With these facts, we now complete the proof of \cref{thm:dudley_bound}.

\begin{proof}[Proof of \cref{thm:dudley_bound}]
Observe that we may assume $\eps \geq 1/m$, as otherwise we can simply return $\G$ as our output sparsifier. Similarly, we may assume $m \ge n$, or else $\G$ itself is a good enough sparsifier. We bound the constant $\gamma$ in \Cref{thm:chaining} using \cref{lemma:dx0,lemma:largen,lemma:mainak17}. For the sets $T_s, T_{s+1}, \dots$ constructed in \cref{lemma:largen,lemma:mainak17}, we have
\begin{align}
\label{eq:gamm}
\gamma \leq \sup_{x \in B} 2^{s/2} \cdot d(x, \vec{0}) + \sum_{N \ge s} 2^{N/2} d(x, T_N).
\end{align}

For the choices $s = \lceil \log_2 \log n \rceil$ and $Z =  2 \log(8 (1+ \Cchain) m)$, write
\[
 \sum_{N \ge s} 2^{N/2} d(x, T_N) =  \sum_{N \in [s, Z]} 2^{N/2} d(x, T_N) + \sum_{N \geq Z} 2^{N/2} d(x, T_N).
\]
\cref{lemma:largen} implies 
\[
\sum_{N \geq Z} 2^{N/2} d(x, T_N) \leq 2 \sum_{N \geq Z} 2^{N/2 - 2^{N-1}/m}  \leq 2 \sum_{N \geq Z} 2^{- 2^{N}/(4m)} \leq \frac{1}{100 \Cchain m} \leq \frac{\eps}{100 \Cchain}.
\]
On the other hand, for $\rho = \Cbd \eps^{-2} \log^3 m $ \cref{lemma:mainak17} implies 
\[
\sum_{N \in [s, Z]} 2^{N/2} d(x, T_N) \leq \sum_{N \in [s, Z]}  \Cf \rho^{-1/2} \sqrt{\log m} \leq \frac{\Cf Z \eps \sqrt{\log m}}{\sqrt{\Cbd} \log^{3/2} m} = \frac{2\Cf \eps \log(8 (1+ \Cchain) m) }{\sqrt{\Cbd} \log m}.
\]
Finally, \cref{lemma:dx0} implies $d(x,0) \leq \rho^{-1/2}$. Plugging these into \cref{eq:gamm} and using $m \geq n$ yields 
\[
\gamma \leq \frac{2 \eps \sqrt{\log m}}{\sqrt{\Cbd} \log^{3/2} m} + \frac{2\Cf \eps \log(8 (1+ \Cchain) m) }{\sqrt{\Cbd} \log m} + \frac{\eps}{100 \Cchain}
\]
As $m \geq 2$ without loss of generality we have $8 (1+ \Cchain) m \leq m^{4+ \log_2(1+ \Cchain)}$: the above yields
\[
\gamma \leq \frac{2 \eps }{\sqrt{\Cbd}} + \frac{(8+ 2\log_2(1+ \Cchain))\Cf \eps}{\sqrt{\Cbd}} + \frac{\eps}{100 \Cchain}.
\]
For $\Cbd = 2 \Cchain^2 (2 + (8+ 2\log_2(1+ \Cchain)) \Cf)^2$ the above gives $\gamma \le \frac{\eps}{\Cchain}$: the result follows from \cref{thm:chaining}.
\end{proof}

\section{Improved Size Bound from Chaining}
\label{sec:chaining}

\newcommand{\Ct}{C_{\hyperref[lemma:triangle]{\Delta}}}
\newcommand{\Cball}{C_{\hyperref[lemma:convex]{5}}}
\newcommand{\Ctal}{C_{\hyperref[thm:enough]{6}}}
\newcommand{\Csc}{C_{{\hyperref[lemma:strongconvex]{7}}}}
\newcommand{\Cgf}{C_{\hyperref[eq:deff]{8}}}
\newcommand{\Clamb}{C_{{\hyperref[lemma:well_separated]{9}}}}
\newcommand{\Cinfc}{C_{{\hyperref[lemma:infcovering]{10}}}}
\newcommand{\Cltwo}{C_{{\hyperref[lemma:l2covering]{11}}}}
\newcommand{\Cstar}{C_{{\hyperref[lemma:cstar]{12}}}}
\newcommand{\Cmain}{C_{{\hyperref[thm:chaining_bound]{13}}}}
In this section, we obtain an improved size bound of $O(n\eps^{-2} \log m \log r)$ using a more sophisticated chaining argument. Before we begin, it is helpful to describe how we differ from the result obtained in \Cref{sec:dudley}. Informally, the analysis of the previous section constructed sets $T_i$ where $\sup_{x \in B} d(x,T_i)$ was sufficiently small. These give a bound on $\gamma$, as 
\begin{align*} \gamma \leq \sup_{x \in B} 2^{s/2} \cdot d(x, \vec{0}) + \sum_{N \ge s} 2^{N/2} d(x, T_N) \leq  2^{s/2} \left( \sup_{x \in B}  d(x, \vec{0})\right)+ \sum_{N \ge s} 2^{N/2} \left( \sup_{x \in B}  d(x, T_N)\right). \end{align*}
Constructing the sets $T_i$ in turn is relatively straightforward, as a simple greedy packing argument reduces the problem to estimating the \emph{entropy numbers}\footnote{The $n^{th}$ entropy number of a set $B$ with respect to distance $d$ is the smallest $\epsilon$ such that there exists a set $T$ with $|T| \leq 2^{2^n}$ and $d(x,T) \leq \eps$ for all $x \in B$. \cref{lemma:mainak17} in fact bounds exactly these entropy numbers, although we do not call them that explicitly.} of $B$ with respect to the distance $d$. 

Unfortunately, the bounds obtained by this technique (first developed in an explicit form by Dudley \cite{Dud67}) are suboptimal in many settings: our approach for bounding $\sup_{x \in B} d(x,T_i)$ is essentially tight (\cref{thm:ak17} is tight up to constant factors in the exponent, as is stated in Theorem 6.1 of \cite{AK17}), and the analysis loses from up to $O(\log m)$ levels of the scale parameter $N$. On the other hand, the expression for $\gamma$ critically takes the supremum over the \emph{sum} of all scales: if only a small number of the values $d(x, T_i)$ can be near the supremum for a fixed $x$, the resulting bound we obtain can be significantly tighter.

Actually exploiting this potential for amortization is challenging however, as doing so seems to require additional geometric structure of the metric distance $d$. In this section we employ a chaining framework of \cite{book} based on \emph{growth functionals}, a powerful technique which uses the geometry of the space of events to control $\gamma$. This framework is based on providing a sequence of functions $F_i$ satisfying a certain growth condition \cref{def:growth}. Our approach in this section mirrors previous applications \cite{R96,book} of the framework in proving matrix concentration bounds for sums of rank-$1$ matrices. However, our setting introduces additional complications beyond the matrix setting, which we briefly discuss here.

A key source of difficulty in applying the technique of \cite{R96,book} is the fact that $f_\G(x)$ is not strongly convex. This strongly differs from the rank-$1$ Chernoff setting, where the sum of rank-$1$ matrices yields $x^\top \left( \sum_i v_i v_i^\top \right) x$, which \emph{is} strongly convex in the matrix norm formed by $\sum_i v_i v_i^\top$. Strong convexity enables us to prove lower bounds on the difference of growth functionals (\cref{lemma:growth1}): in the rank-1 matrix case this property immediately allows us to obtain the optimal sparsity bounds. Without this property (as noted in \cite{book}), the growth functional framework seems to break down at $\rho = \eps^{-2} \log^2 m$.  While we have access to a natural matrix $\mA^\top \mW \mA$ to perform the analysis in, it unfortunately does not approximate $f_\G(x)$ well enough for our purposes: for some vectors $x$ we may have $x^\top \mA^\top \mW \mA x \ll f_\G(x)$. Thus, we perform our analysis in a ``mixed'' $(w+f)$-norm (\cref{def:norm}) which contains both $\mA^\top \mW \mA$ and the energy $f_{\G}(x)$ and establish a strong convexity bound (\cref{lemma:strongconvex}) which suffices for our purposes. 

A secondary issue related to strong convexity arises from the distance function $d$ defined in the previous section. A feature of the growth functional framework its use of ``well-separated'' sets (\cref{def:separated}), which have small $d$-diameter but are in some sense ``far apart'' under $d$. However, the analysis of our growth functional requires a stronger property: namely, that convex hulls of the well-separated sets have small $d$-diameter. If $d(x,\cdot)$ (for every fixed $x$) were a convex function, this fact would hold immediately: however we believe that convex combinations of points may grow the $d$-distance arbitrarily. To avoid this issue, we introduce a carefully designed proxy distance function $\wh{d}$ which overestimates $d$. We show that while $\wh{d}(x,\cdot)$ is still not convex (and in fact does not satisfy the triangle inequality) it has these properties in an approximate sense (\Cref{lemma:triangle,lemma:convex}) which suffices for our analysis.

We state the main technical result of this section.

\begin{theorem}
\label{thm:chaining_bound}
Let $\G = (\S, \mA)$ be a unit hypergraph, and let $\tau$ be given group leverage scores (\cref{def:over}) with valid weights $w$ (which do not need to be known). For any constant $C$, there is an absolute constant $\Cmain$ (which depends on $C$) such that $\H = \Subsample(\G, \tau, \rho)$ (\cref{alg:subsample}) with $\rho = \Cmain \eps^{-2}\log m \log r$ satisfies with probability at least $1 - n^{-C}$ the following:
\[
(1-\eps)f_{\G}(x) \le f_{\H}(x) \le (1+\eps)f_{\H}(x)
\text{ for all }
x \in \R^n\,
\]
\end{theorem}
As discussed, we introduce a new distance function with favorable ``convexity'' properties.
\begin{definition}[Modified Distance]
\label{def:modified}
For a fixed hypergraph $\G = (\S, \mA)$ with sampling probabilities $p_i \in (0,1]$, we define $\wh{d} : \R^n \times \R^n \rightarrow \R$ for all $x, y \in \R^n$ as
\begin{align} \wh{d}(x, y) \defeq \left(\sum_{i \in [k]} \one_{\{p_i \neq 1\}} p_i^{-1} (g_i(x) + g_i(y))^2 g_i(x-y)^2 \right)^{1/2}. \label{eq:d2} 
\end{align}
We use the notation $B_{\wh{d}}(x,r) = \{y : y \in \R^n, \wh{d}(x,y) \leq r \}$ to denote balls of radius $r$ in the metric $\wh{d}$.
\end{definition}
We observe that $d(x, y) \le \wh{d}(x, y)$ because $|g_i(x) - g_i(y)| \le g_i(x-y)$. However, our modified distance function no longer satisfies the triangle inequality and thus does not induce a metric. We will instead show that $\wh{d}$ has metric-like properties and we show that this relaxation is unproblematic for the chaining framework.
\subsection{Metric Properties of \texorpdfstring{$\wh{d}$}{wtd}}
\label{subsec:metric}
We first show that $\wh{d}$ satisfies the triangle inequality up to constants. Throughout the remainder of the paper, we do not optimize our constants.
\begin{lemma}
\label{lemma:triangle}
For all $x, y, z \in \R^n$, we have $\wh{d}(x, y) \leq \Ct \left( \wh{d}(x, z) + \wh{d}(y, z) \right)$, where $\Ct \geq 1$ is an absolute constant.
\end{lemma}
\begin{proof}
It suffices to show that $\wh{d}(x, y)^2 \leq \Ct^2 \left( \wh{d}(x, z)^2 + \wh{d}(y, z)^2 \right)$, as $\sqrt{a^2+b^2} \leq a+b$ for nonnegative $a,b$. We show this by proving that for all $i \in [k]$,
\begin{align} 
(g_i(x) + g_i(y))^2 g_i(x-y)^2 \leq \Ct^2 \left((g_i(x) + g_i(z))^2 g_i(x-z)^2\right) + \left((g_i(y) + g_i(z))^2 g_i(y-z)^2\right). \label{eq:need} 
\end{align}
Without loss of generality, assume $g_i(x) \ge g_i(y)$. We split into cases based on how large $g_i(z)$ is. If $g_i(z) \ge g_i(x)/2$, then we have (noting $g_i(y) \geq 0$)
\begin{align*}
    \left((g_i(x) + g_i(z))^2 g_i(x-z)^2\right) + \left((g_i(y) + g_i(z))^2 g_i(y-z)^2\right) &\geq \frac{1}{4} \left( g_i(x)^2(g_i(x-z)^2 + g_i(y-z)^2) \right)\\
    &\overset{(i)}{\geq} ~\frac{1}{8} g_i(x)^2 g_i(x-y)^2 \\
    &\geq \frac{1}{32} (g_i(x) + g_i(y))^2 g_i(x-y)^2,
\end{align*}
where $(i)$ uses $(a+b)^2 \leq 2a^2 + 2b^2$ and that $g_i(x-z) + g_i(z-y) \geq g_i(x-y)$.
In the other case, $g_i(z) \le g_i(x)/2$, so $g_i(x-z) \ge g_i(x) - g_i(z) \ge g_i(x)/2$ by the triangle inequality on $g_i$. Now,
\[ 
(g_i(x) + g_i(y))^2 g_i(x-y)^2 \leq  (g_i(x) + g_i(y))^4 \leq 16 g_i(x)^4 
\]
by the triangle inequality on $g_i$, and
\[ 
(g_i(x) + g_i(z))^2 g_i(x-z)^2 \geq \frac{1}{4} g_i(x)^4.
\] 
The result follows by choosing $\Ct = 8$.
\end{proof}

The other property we need is that balls in $\wh{d}$ are approximately convex. This is the reason we introduce $\wh{d}$: while $d$ satisfies the triangle inequality we believe that it \emph{does not} satisfy this additional property. This is a major reason we use $\wh{d}$ as opposed to $d$ in this section.
\begin{lemma}
\label{lemma:convex}
Let $\hull(S)$ denote the convex hull of a set $S$. For all $x \in \R^n$ and $a \ge 0$, we have $\hull(B_{\wh{d}}(x, a)) \subseteq B_{\wh{d}}(x, \Cball a)$ for a universal constant $\Cball$.
\end{lemma}
\begin{proof}
Let $x_1, \dots, x_t \in B_{\wh{d}}(x, a)$, so $\wh{d}(x, x_j) \le a$ for $j \in [t]$.
Consider the convex combination $y = \sum_{j\in[t]} c_j x_j$, where $\sum_{j\in[t]} c_j = 1$ and all $c_j \ge 0$. We wish to show $\wh{d}(x, y) \le \Cball a$. Because $g_i(\cdot)$ is positive and convex for all $i \in [k]$, note that both $g_i^2$ and $g_i^4$ are convex functions.

We first write 
\begin{align} \wh{d}(x, y)^2 &= \sum_{i \in [k]} \one_{\{p_i \neq 1\}} p_i^{-1} (g_i(x) + g_i(y))^2 g_i(x-y)^2 \nonumber \\
&\le 2\sum_{i \in [k]} \one_{\{p_i \neq 1\}} p_i^{-1} g_i(x)^2 g_i(x-y)^2 + 2\sum_{i \in [k]} \one_{\{p_i \neq 1\}} p_i^{-1} g_i(y)^2 g_i(x-y)^2. \label{eq:boundstuff} \end{align}
For the first term, convexity of $g_i^2$ and nonnegativity of $g_i$ implies
\begin{align}
    \sum_{i \in [k]} \one_{\{p_i \neq 1\}} p_i^{-1} g_i(x)^2 g_i(x-y)^2 &\le \sum_{i \in [k]} \one_{\{p_i \neq 1\}} p_i^{-1} g_i(x)^2 \sum_{j\in[t]} c_jg_i(x-x_j)^2 \nonumber \\
    &\le \sum_{j \in [t]} c_j \sum_{i\in[k]} \one_{\{p_i \neq 1\}} p_i^{-1} (g_i(x) + g_i(x_j))^2 g_i(x - x_j)^2 \nonumber \\
    &= \sum_{j \in [t]} c_j \wh{d}(x, x_j)^2 \le a^2. \label{eq:boundfirstterm}
\end{align}
We turn our attention to the second term in \eqref{eq:boundstuff}. To this end define the following sets of ``big coordinates'': let $\B \defeq \{i \in [k] : g_i(y) \ge 3 g_i(x)\}$, and for each $j \in [t]$ let $\B_j \defeq \{i \in [k] : g_i(x_j) \ge 2 g_i(x)\}$. It is important that the constant $3$ in $\B$ is bigger than the $2$ in $\B_j$.
First observe by \eqref{eq:boundfirstterm}
\begin{align}
\label{eq:small_inds}
\sum_{i \in [k]\backslash \B} \one_{\{p_i \neq 1\}} p_i^{-1} g_i(y)^2 g_i(x-y)^2 \le 9 \sum_{i \in [k] \backslash \B} \one_{\{p_i \neq 1\}} p_i^{-1} g_i(x)^2 g_i(x-y)^2 \le 9a^2. 
\end{align}
On the other hand, for any $i \in \B$, note that $g_i(x-y) \le g_i(x) +g_i(y) \leq 4/3 \cdot g_i(y)$. Thus,
\begin{align}
\sum_{i \in \B} \one_{\{p_i \neq 1\}} p_i^{-1} g_i(y)^2 g_i(x-y)^2 &\le (4/3)^2 \sum_{i \in \B} \one_{\{p_i \neq 1\}} p_i^{-1} g_i(y)^4 \nonumber \\
&\le (4/3)^2 \sum_{i \in \B} \one_{\{p_i \neq 1\}} p_i^{-1} \sum_{j\in[t]} c_j g_i(x_j)^4
\label{eq:gy4}   
\end{align} 
where the last line used the convexity of $g_i^4$.
To bound \cref{eq:gy4}, we turn our attention to the $x_j$ and the sets $\B_j$. We begin by showing for any $i \in \B$ 
\begin{align}
\sum_{j\in[t]} c_j g_i(x_j)^4 \leq \frac{5}{4} \sum_{j \in [t] : \B_j \owns i} c_j g_i(x_j)^4.
\label{eq:desired}
\end{align}
To see this, note that for $i \in \B$ we have  $g_i(y) \ge 3g_i(x)$: therefore
\begin{align*}
\sum_{j \in [t]: \B_j \not\owns i} c_j  g_i(x_j)^4 \le 16g_i(x)^4 \le \frac15 g_i(y)^4 \le \frac{1}{5} \sum_{j \in [t]} c_j g_i(x_j)^4.
\end{align*}
Thus \eqref{eq:desired} follows by rearranging
\begin{align*}
\sum_{j \in [t]} c_j g_i(x_j)^4 = \sum_{j \in [t] : \B_j \owns i} c_j g_i(x_j)^4 + \sum_{j \in [t] : \B_j \not\owns i} c_j g_i(x_j)^4 \le \sum_{j \in [t] : \B_j \owns i} c_j g_i(x_j)^4  + \frac{1}{5} \sum_{j \in [t]} c_j g_i(x_j)^4.
\end{align*}
Now for $i \in \B_j$ we have $g_i(x-x_j) \ge g_i(x_j) - g_i(x) \ge g(x_j)/2$.  Because $\wh{d}(x, x_j) \le a$ for all $j \in [t]$, we get
\[ a^2 \ge \wh{d}(x, x_j)^2 = \sum_{i\in[k]} \one_{\{p_i \neq 1\}} p_i^{-1} (g_i(x) + g_i(x_j))^2 g_i(x - x_j)^2  \ge \frac14 \sum_{i \in \B_j}\one_{\{p_i \neq 1\}} p_i^{-1} g_i(x_j)^4. \]
Thus, reweighting by $c_j$ and summing gives
\begin{align} \sum_{j \in [t]} \sum_{i \in \B_j} c_j \one_{\{p_i \neq 1\}} p_i^{-1} g_i(x_j)^4 \le 4a^2. \label{eq:abovestuff} \end{align}
To finish, we write 
\begin{align}
\label{eq:gy4_ub}
\sum_{i \in \B} \left[ \sum_{j \in [t]} c_j  \one_{\{p_i \neq 1\}} p_i^{-1} g_i(x_j)^4 \right] &\overset{(i)}{\le} \frac{5}{4} \sum_{i \in \B} \sum_{j \in [t]: \B_j \owns i} c_j \one_{\{p_i \neq 1\}} p_i^{-1} g_i(x_j)^4 \nonumber \\
&\overset{(ii)}{=} \frac{5}{4} \sum_{j \in [t]} \sum_{i \in \B_j \cap \B} c_j \one_{\{p_i \neq 1\}} p_i^{-1} g_i(x_j)^4  \overset{(iii)}{\leq} 5a^2
\end{align}
where $(i)$ used \eqref{eq:desired}, $(ii)$ swaps the order of summation, and $(iii)$ uses \eqref{eq:abovestuff}. Putting this together, we have
\begin{align*}
    \wh{d}(x,y)^2 &\overset{(i)}{\leq} 2 \sum_{i \in [k]} \one_{\{p_i \neq 1\}} p_i^{-1} g_i(x)^2 g_i(x-y)^2 + 2 \sum_{i \in [k]} \one_{\{p_i \neq 1\}} p_i^{-1} g_i(y)^2 g_i(x-y)^2 \\
    &\overset{(ii)}{\leq} 2a^2 + 2 \sum_{i \in [k]} \one_{\{p_i \neq 1\}} p_i^{-1} g_i(y)^2 g_i(x-y)^2 \\
    &= 2a^2 + 2 \sum_{i \in [k] \backslash \B} \one_{\{p_i \neq 1\}} p_i^{-1} g_i(y)^2 g_i(x-y)^2 + 2 \sum_{i \in \B} \one_{\{p_i \neq 1\}} p_i^{-1} g_i(y)^2 g_i(x-y)^2  \\
    &\overset{(iii)}{\leq} 2a^2 + 18a^2 + \frac{32}{9} \sum_{i \in \B} \one_{\{p_i \neq 1\}} p_i^{-1} g_i(y)^4 \overset{(iv)}{\leq} 20a^2 + \frac{160}{9} a^2 \leq 40 a^2,
\end{align*}
where $(i)$ used \cref{eq:boundstuff}, $(ii)$ used \cref{eq:boundfirstterm}, $(iii)$ used \cref{eq:small_inds} and \cref{eq:gy4}, and $(iv)$ used \cref{eq:gy4_ub}. The claim follows by choosing $\Cball = \sqrt{40}$. 
\end{proof}

\subsection{Chaining Functions and the Growth Condition}
\label{subsec:growth}
Our next goal is to introduce the \emph{growth function} framework for chaining. We closely follow the presentation in \cite{book}, which in turn is based on the approach of \cite{R96}.  To start, we must introduce well-separated sets.
\begin{definition}[Well-separated sets]
\label{def:separated}
Given $a > 0$ and integer $\lambda \ge 4$, we say that subsets $H_1, \dots, H_t \subseteq \R^n$ are $(a,\lambda)$-separated for a distance $\wh{d}$ if there are points $x_1, \dots, x_t$ such that $\wh{d}(x_i, x_j) \ge a$ for all $i \neq j$ and $H_i \subseteq B_{\wh{d}}(x_i, a/\lambda)$. 
\end{definition}
In our application $\lambda$ will be chosen to be a sufficiently large constant.

We will now define what it means for certain functions to satisfy a \emph{growth condition}. Recall that $B$ is the unit ball of the energy $f_{\G}(x)$.
\begin{definition}[Growth condition]
\label{def:growth}
We say that functions $F_0, F_1, \dots$ that take as inputs subsets of $\R^n$, and output nonnegative real numbers satisfy a \emph{growth condition} with parameters $c^*$ and $\lambda$ if for all $N \ge 0, a > 0$, and $(a,\lambda)$-separated $H_1, \dots, H_{2^{2^N}} \subseteq B$,
\begin{align} F_N\left(\bigcup_{i=1}^{2^{2^N}} H_i \right) \ge c^* a 2^{N/2} + \min_{1 \le i \le 2^{2^N}} F_{N+1}(H_i). \label{eq:growth} \end{align}
We also require that the functions are \emph{decreasing}: $F_{N+1}(H) \le F_N(H)$ for all $H \subseteq B$ and $N \ge 0$.
\end{definition}
While this definition may be unintuitive, it turns out that constructing such functions immediately implies bounds on chaining.
\begin{theorem}[Theorem 2.3.16 in \cite{book}]
\label{thm:enough}
If $\wh{d}$ satisfies the triangle inequality up to a constant (in the sense of \cref{lemma:triangle}), and $F_0, F_1, \dots$ are a decreasing sequence of functions satisfying the condition of \cref{def:growth}, then for an absolute constant $\Ctal$ (which depends on $\Ct$)
\[ \inf_{\substack{T_0, T_1, \dots \\ T_N \subseteq B, |T_N| \le 2^{2^N} \text{ for } N \ge 0}} \sup_{x \in B} \sum_{N \ge 0} 2^{N/2} \wh{d}(x, T_N) \leq \Ctal \left( \frac{\lambda F_0(B)}{c^*} + \lambda \mathrm{diam}_{\wh{d}}(B) \right). \]
Here, $ \mathrm{diam}_{\wh{d}}(B) = \max_{x, y\in B} \wh{d}(x,y)$.
\end{theorem}
Formally, Theorem 2.3.16 in \cite{book} is only stated when $\wh{d}$ is a metric. However, one can check (and it is stated on page 565 of \cite{book}) that \cref{thm:enough} works as long as $\wh{d}$ satisfies the triangle inequality up to a constant, as was shown in \cref{lemma:triangle}. Also, in \cite[Definition 2.3.8]{book}, well-separated has an additional property that all the sets $H_i$ are contained inside a ball of radius $\lambda a$, which we do not include since we do not need it. Because our definition of well-separated is more general, \cref{thm:enough} is still true.

\subsection{Constructing Chaining Functions}
\label{subsec:construct}
The goal of this section is to define functions $F_N$ and show that they satisfy the growth condition, \cref{def:growth}. It is useful to first define a norm related to the energy. In this section, $\tau$ are group leverage score overestimates with corresponding weights $w$ (in the sense of \cref{def:over}).
\begin{definition}[$\|\cdot\|_{w+f}$-norm]
\label{def:norm}
For $x \in \R^n$ we define $\|x\|_{w+f} \defeq \sqrt{\frac{1}{2}\left(\|x\|_{\mA^\top \mW \mA}^2 + f_{\G}(x)\right)}$, where
$\|x\|_{\mA^\top \mW \mA} \defeq \sqrt{x^\top \mA^\top \mW \mA x}$.
\end{definition}
Note that $\|cx\|_{w+f} = |c|\|x\|_{w+f}$ for all $c \in \R$. We now verify a ``strong-convexity'' property of $\|\cdot\|_{w+f}$. 
\begin{lemma}
\label{lemma:strongconvex}
There is an absolute constant $\Csc$ such that for all $x, y$ with $\max\{\|x\|_{w+f}, \|y\|_{w+f}\} \le 1$:
\[ 
1 - \left\|\frac{x+y}{2}\right\|_{w+f} \geq \Csc \left( \|x-y\|_{\mA^\top \mW \mA}^2 + \sum_{i \in [k]} (g_i(x) - g_i(y))^2 \right)
\]
\end{lemma}
\begin{proof}
Note that $\left\|\frac{x+y}{2}\right\|_{\mA^\top \mW \mA}^2 = \frac{1}{2}\left(\|x\|_{\mA^\top \mW \mA}^2 + \|y\|_{\mA^\top \mW \mA}^2\right) - \left\|\frac{x-y}{2}\right\|_{\mA^\top \mW \mA}^2$, and
\begin{align*}
    f_{\G}\left(\frac{x+y}{2}\right) &= \sum_{i \in [k]} g_i\left(\frac{x+y}{2}\right)^2 \overset{(i)}{\le} \sum_{i \in [k]} \left(\frac{1}{2}g_i(x) + \frac{1}{2}g_i(y)\right)^2 \\
    &= \sum_{i \in [k]} \frac{1}{2}g_i(x)^2 + \frac12g_i(y)^2 - \left(\frac{1}{2}g_i(x) - \frac{1}{2}g_i(y) \right)^2 \\
    &= \frac{1}{2}\left(f_{\G}(x) + f_{\G}(y)\right) - \sum_{i \in [k]} \left(\frac{1}{2}g_i(x) - \frac{1}{2}g_i(y) \right)^2.
\end{align*}
where inequality $(i)$ follows from the convexity of $g_i$. Summing these equations and dividing by $2$ yields
\begin{align*} \left\|\frac{x+y}{2}\right\|_{w+f}^2 &\le \frac{1}{2}\Bigg(\frac{1}{2}\left(\|x\|_{\mA^\top \mW \mA}^2 + \|y\|_{\mA^\top \mW \mA}^2\right) - \left\|\frac{x-y}{2}\right\|_{\mA^\top \mW \mA}^2 \\ &+ \frac{1}{2}\left(f_{\G}(x) + f_{\G}(y)\right) - \sum_{i \in [k]} \left(\frac{1}{2}g_i(x) - \frac{1}{2}g_i(y) \right)^2 \Bigg) \\
&\le \frac12\|x\|_{w+f}^2 + \frac12\|y\|_{w+f}^2 - \frac12\left\|\frac{x-y}{2}\right\|_{\mA^\top\mW\mA}^2-\frac12 \sum_{i\in[k]}\left(\frac12 g_i(x) - \frac12g_i(y)\right)^2
\\
&\le 1 - \frac{1}{8}\left(\|x-y\|_{\mA^\top \mW \mA}^2 + \sum_{i\in[k]} (g_i(x) - g_i(y))^2 \right), \end{align*}
where we used that $\|x\|_{w+f}, \|y\|_{w+f} \le 1$. Taking squares roots and using $\sqrt{1-t/8} \le 1-t/16$ for all $t \le 8$ gives us
\[ \left\|\frac{x+y}{2}\right\|_{w+f} \le 1 - \frac{1}{16}\left(\|x-y\|_{\mA^\top \mW \mA}^2 + \sum_{i\in[k]} (g_i(x) - g_i(y))^2 \right) \] which rearranges to the desired inequality with $\Csc = 1/16$.
\end{proof}
We are now ready to define the functions $F_N$. Let
\begin{align}
F_N(A) = 1 - \inf_{x\in \hull(A)} \|x\|_{w+f} + \frac{1}{\log m} \max\{0, 1 + \Cgf \log m - N\}, \label{eq:deff}
\end{align}
for a sufficiently large constant $\Cgf$ specified later. To explain the intuition for the last term, note that handling terms with $N \ge \Cgf \log m$ is simple, as we may employ a similar technique to \cref{lemma:largen}. We formalize this later in \cref{lemma:cstar,lemma:alln}: the former shows $F$ satisfies the growth condition for $N \leq \Cgf \log m$, and the latter shows a slight modification of $F$ satisfies it for all $N$. 

\begin{lemma}
\label{lemma:growth1}
Let $N \leq \Cgf \log m$ and fix $S_N = 2^{2^N}$. Let $H_1, \dots, H_{S_N}$ be a collection of $(a, \lambda)$-separated sets, and define $v_i = \argmin\{\|x\|_{w+f} : x \in \hull(H_i)\}$. Then for
\begin{align} 
R^2 \defeq \max_{j, j' \in [S_N]} \left( \|v_j - v_{j'}\|_{\mA^\top \mW \mA}^2 + \sum_{i\in[k]} (g_i(v_j) - g_i(v_{j'}))^2\right). \label{eq:defr} 
\end{align}
we have
\[
F_N\left(\bigcup_{i=1}^{S_N} H_i \right) \geq \left(\Csc R^2 + \frac{1}{\log m}\right) + \min_{1 \le i \le S_N} F_{N+1}(H_i)
\]
\end{lemma}
\begin{proof}
Let $u = \max_{i\in[S_N]} \|v_i\|_{w+f}$.
Note that then $\|v_i/u\|_{w+f} \le 1$ for all $i \in [S_N]$.

Now we know that for all $j, j' \in [S_N]$, $(v_j+v_{j'})/2 \in \hull(\cup_{i\in[S_N]} H_i)$. Thus for all $j, j' \in [S_N]$ we get 
\begin{align} &F_N\left(\bigcup_{i=1}^{S_N} H_i \right) - \min_{1 \le i \le S_N} F_{N+1}(H_i) \ge u - \left\|\frac{v_j+v_{j'}}{2}\right\|_{w+f} + \frac{1}{\log m} \nonumber \\
= ~&u\left(1 - \left\|\frac{v_j/u+v_{j'}/u}{2}\right\|_{w+f}\right) + \frac{1}{\log m} \nonumber \\
\overset{(i)}{\ge} ~& \Csc u\left(\|v_j/u - v_{j'}/u\|_{\mA^\top \mW \mA}^2 + \sum_{i\in[k]} (g_i(v_j/u) - g_i(v_{j'}/u))^2 \right) + \frac{1}{\log m}  \nonumber \\
= ~&\Csc u^{-1}\left(\|v_j - v_{j'}\|_{\mA^\top \mW \mA}^2 + \sum_{i\in[k]} (g_i(v_j) - g_i(v_{j'}))^2 \right) + \frac{1}{\log m} \nonumber \\
\overset{(ii)}{\ge} ~& \Csc \left(\|v_j - v_{j'}\|_{\mA^\top \mW \mA}^2 + \sum_{i\in[k]} (g_i(v_j) - g_i(v_{j'}))^2 \right) + \frac{1}{\log m}. \label{eq:boundrlogm} \end{align}
Here $(i)$ follows from \cref{lemma:strongconvex}, and $(ii)$ follows from $u \le 1$, as $\|x\|_{w+f} \le 1$ for any $x \in B$ by \cref{lemma:xawax}. The claim follows by choosing $j, j'$ to maximize the right-hand side.
\end{proof}

To get a complete proof of \eqref{eq:growth}, we need to relate $R$ in \cref{lemma:growth1} to $a$, which is the separation of the clusters. To achieve this, we first prove that the approximate metric properties of $\wh{d}$ (\cref{lemma:triangle,lemma:convex}) combined with well-separatedness (\cref{def:separated}) imply that all points in different $\hull(H_i)$ are far apart under $\wh{d}$. 

\begin{lemma}
\label{lemma:well_separated}
Let $H_1, \dots, H_k$ be a collection of $(a, \lambda)$-separated sets (\cref{def:separated}). There is an absolute constant $\Clamb$ such that for any $\lambda \geq \Clamb$ and any distinct $i, j \in [k]$ and any $y_i \in \hull(H_i), y_j \in \hull(H_j)$
\[
\wh{d}(y_i, y_j) \geq a/(2 \Ct^2)\,.
\]
\end{lemma}
\begin{proof}
Let $x_i, x_j$ be the center points of $H_i$ and $H_j$ defined in \cref{def:separated}. We observe by \cref{lemma:triangle} and nonnegativity of $\wh{d}$ 
\[
\wh{d}(x_i, x_j) \leq \Ct \left( \wh{d}(x_i, y_i) + \wh{d}(y_i, x_j) \right) \leq \Ct^2 \wh{d}(x_i, y_i) + \Ct^2 \wh{d}(y_j, x_j) + \Ct^2 \wh{d}(y_i, y_j).
\]
Rearranging the above yields
\begin{align*}
\wh{d}(y_i, y_j) &\geq \Ct^{-2} \wh{d}(x_i, x_j) - \left(\wh{d}(x_i, y_i) + \wh{d}(x_j, y_j) \right) \\
&\overset{(i)}{\geq} a/\Ct^{2} - \left(\max_{z \in \hull(H_i)} \wh{d}(x_i, z) + \max_{z \in \hull(H_j)} \wh{d}(x_j, z)\right)  \\
&\overset{(ii)}{\ge} a/\Ct^{2} - 2 \Cball a / \lambda. 
\end{align*}
Inequality $(i)$ holds via the separation bound on $x_i$ from \cref{def:separated}, and $(ii)$ holds since $ \hull(H_i) \subseteq \hull(B_{\wh{d}}(x_i, a/\lambda)) \subseteq B_{\wh{d}}(x_i, \Cball a/\lambda)$ by \cref{def:separated} and \cref{lemma:convex} respectively. The claim follows as $\lambda \geq \Clamb \defeq 4 \Ct^2 \Cball$.
\end{proof}

Our goal now is to partition the points $v_1, \dots, v_{S_N}$ defined in \cref{lemma:growth1} into less than $S_N$ groups, and upper bound the maximum distance between points in the same group. This upper bound will be in terms of $R$. Because the pigeonhole principle ensures that some two points will be in the same group, combining this with \cref{lemma:well_separated} gives a lower bound on $R$ in terms of $a$: this will allow us to prove \cref{eq:growth}. To understand how to build such a partition, we will prove an upper bound on $\wh{d}$ and analyze the terms of this bound separately. 
\begin{lemma}
\label{lemma:d_bound}
Let $z\in \R^n$ be a fixed vector. For any $x,y \in \R^n$, define 
\[
\wh{d}_z(x,y) = \sum_{i \in [k]} \one_{\{p_i \neq 1\}} p_i^{-1} g_i(x-y)^2 g_i(z)^2
\enspace\text{ and }\enspace
\wh{d}^\infty(x,y) = \max_{i \in [k]} \one_{\{p_i \neq 1\}} p_i^{-1} g_i(x-y)^2 
\,.
\]
We have
\begin{align}
\wh{d}(x, y)^2 &\leq 3 \wh{d}^\infty(x,y) \sum_{i \in [k]} (g_i(y)-g_i(z))^2 + 3  \wh{d}^\infty(x,y) \sum_{i \in [k]} (g_i(x)-g_i(z))^2 \label{eq:boundd1} \\
&\phantom{=} + 12 \wh{d}_z(x,y). \label{eq:boundd2}
\end{align}
\end{lemma}
\begin{proof} 
We have
\begin{align*}
    \wh{d}(x, y)^2 &= \sum_{i \in [k]} \one_{\{p_i \neq 1\}} p_i^{-1} (g_i(x) + g_i(y))^2 g_i(x-y)^2  \\
    &= \sum_{i \in [k]} \one_{\{p_i \neq 1\}} p_i^{-1} g_i(x-y)^2 (g_i(x)-g_i(z)+g_i(y)-g_i(z)+2g_i(z))^2  \\
    &\overset{(i)}{\leq} 3\sum_{i \in [k]} \one_{\{p_i \neq 1\}} p_i^{-1} g_i(x-y)^2 (g_i(x)-g_i(z))^2  + 3 \sum_{i \in [k]} \one_{\{p_i \neq 1\}} p_i^{-1} g_i(x-y)^2 (g_i(y)-g_i(z))^2 \\ &+ 12 \sum_{i \in [k]} \one_{\{p_i \neq 1\}} p_i^{-1} g_i(x-y)^2 g_i(z)^2  \\
    &\overset{(ii)}{\le} 3 \wh{d}^\infty(x,y) \sum_{i \in [k]} (g_i(y)-g_i(z))^2 + 3  \wh{d}^\infty(x,y) \sum_{i \in [k]} (g_i(x)-g_i(z))^2 + 12 \wh{d}_z(x,y)
\end{align*}
where $(i)$ follows from the scalar inequality $(a+b+c)^2 \leq 3 (a^2 + b^2 + c^2)$ and $(ii)$ follows from H\"{o}lder's ineuality for the $\ell_1$-$\ell_\infty$ norms.
\end{proof}
We now construct a small collection of sets which ensure the terms in \eqref{eq:boundd1} above are small in their interior. The next lemma follows easily from \cref{thm:ak17}.
\begin{lemma}
\label{lemma:infcovering}
Let $\G = (\S, \mA)$ be a hypergraph with group leverage score overestimates $\tau$ and weights $w$ (\cref{def:over}). Let $v \in \R^n$ be given, and define $B_w \defeq \{x \in \R^n : \|x-v\|_{\mA^\top\mW\mA} \le R\}$ for $R > 0$. $B_w$ can be covered with at most $S_{N-2}= 2^{2^{N-2}}$ subsets $P_1, \dots, P_{S_{N-2}}$ such that for an absolute constant $\Cinfc$
\[ 
\max_{\substack{j \in [S_{N-2}] \\ x, y \in P_j}} \wh{d}^\infty(x,y)  \le \Cinfc 2^{-N} \rho^{-1}R^2 \log m\,.
\]
\end{lemma}
\begin{proof}
By shifting and scaling, we can assume that $R = 1$ and $v = 0$. Analogously to \cref{lemma:mainak17}, define $u_j = \tau_i^{-1/2}(\mA^\top\mW\mA)^{-1/2}a_j$ for $j \in S_i$: note that
\[ 
\l a_j, x\r = \l (\mA^\top\mW\mA)^{-1/2}a_j, (\mA^\top\mW\mA)^{1/2}x \r = \tau_i^{1/2} \l u_j, (\mA^\top\mW\mA)^{1/2}x\r 
\] 
and $\|(\mA^\top\mW\mA)^{1/2}x\|_2 \le 1$ for $x \in B_w$. Also, $\|u_j\|_2 = \tau_i^{-1} a_j^\top (\mA^\top\mW\mA)^{-1} a_j \le 1$ by \cref{def:over}.
Let $\eta$ satisfy $m^{\Cak/\eta^2} = S_{N-2}$, so $\eta = \sqrt{\Cak} 2^{1-N/2}\sqrt{\log_2 m}$, and let $P_1, \dots, P_{S_{N-2}}$ be  the subsets given by \cref{thm:ak17} for the vectors $u_i$ and parameter $\eta$.

For any $x, y \in P_{\ell}$ and any $j \in S_i$,
\begin{align} \l a_j, x-y\r^2 = \tau_i \l u_j, (\mA^\top\mW\mA)^{1/2}(x-y) \r^2 \le \tau_i \eta^2. \label{eq:akguarantee} \end{align}
Thus we conclude that for all $x, y \in P_{\ell}$ for $\ell \in [S_{N-2}]$
\[\wh{d}^\infty(x,y) \le \one_{\{p_i \neq 1\}} p_i^{-1} \tau_i \eta^2 \le 4 \Cak (\log 2)^{-1} \rho^{-1} 2^{-N} \log m. \]
The conclusion follows if $\Cinfc = 4 \Cak (\log 2)^{-1}$, where the extra factor of $R^2$ comes the scaling of $B_w$ by a factor of $R$.
\end{proof}
Our next goal is to upper bound \eqref{eq:boundd2}. For this, we will use critically that $z$ is a fixed vector independent of $v_j, v_{j'}$. This allows us to show a covering result that depends on $z$.
\begin{lemma}
\label{lemma:l2covering}
In the setting of \cref{lemma:infcovering} and any $z \in \R^n$, $B_w$ can be covered with at most $S_{N-2}$ subsets $Q_1, \dots, Q_{S_{N-2}}$ such that for an absolute constant $\Cltwo$,
\[ 
\max_{\substack{j \in [S_{N-2}] \\ x, y \in Q_j}} \wh{d}_z(x,y) \le \Cltwo 2^{-N}\rho^{-1} R^2 \log r. \]
\end{lemma}

If we use the same partition $P_1, \dots, P_{S_{N-2}}$ in \cref{lemma:infcovering} as $Q_1, \dots, Q_{S_{N-2}}$ for \cref{lemma:l2covering}, the bound in \cref{lemma:l2covering} would have a $\log m$ term instead of $\log r$.
\begin{proof}

We closely follow the proof of \cref{thm:ak17} in \cite{AK17}. However, unlike its previous applications \cref{lemma:mainak17,lemma:infcovering} we cannot apply \cref{thm:ak17} directly, so we instead adapt the proof method. By shifting and scaling, we can assume that $v = 0$ and $R = 1$.
Note that if $\|x\|_2 \le 1$, then $(\mA^\top \mW \mA)^{-1/2}x \in B_w$. For $\|x\|_2 \le 1$, define the ball \[ B^x \defeq \left\{y \in \R^n : \wh{d}_z(x,y) \le \zeta\right\}, \] where we will choose the parameter $\zeta$ later. For a variance parameter $\sigma^2 \le 1$, define the \emph{Gaussian measure} of a set $A$ as
\[ \gamma_{\sigma}(A) \defeq \frac{1}{(\sqrt{2\pi\sigma^2})^n} \int_A e^{-\|x\|_2^2/(2\sigma^2)} dx. \]
Our first goal is to lower bound $\gamma_{\sigma}(B^x)$. Note that $B^x = x + B^0$, and $B^0$ is a symmetric set. Because $\|(\mA^\top \mW \mA)^{-1/2}a_j\|_2 \le \tau_i^{1/2}$ for all $j \in S_i$ by \cref{def:over}, the Gaussian tail bound implies
\[ \Pr_{y \sim \gamma_{\sigma}}\left[|\l a_j, (\mA^\top \mW \mA)^{-1/2}y\r| \ge t \right] \le e^{-t^2/(2\tau_i\sigma^2)}. \]
By a union bound over all $j \in S_i$ and integration by parts, we conclude
\begin{align*}
&\E_{y \sim \gamma_{\sigma}}[g_i((\mA^\top \mW \mA)^{-1/2}y)^2] \le 4\tau_i \sigma^2 \log r + \int_{2\tau_i^{1/2}\sigma\sqrt{\log r}}^\infty 2t \cdot (r \cdot e^{-t^2/(2\tau_i\sigma^2)}) dt \\
= ~& 4\tau_i \sigma^2 \log r + 2\tau_i \sigma^2 e^{-2 \log r} r \le 6\tau_i \sigma^2 \log r.
\end{align*}
Thus, we get that
\begin{align*}
    ~&\E_{y \sim \gamma_{\sigma}}\left[\sum_{i\in[k]} \one_{\{p_i \neq 1\}} p_i^{-1} g_i((\mA^\top \mW \mA)^{-1/2}y)^2 g_i(z)^2 \right] \\
    \le ~&\sum_{i\in[k]} \one_{\{p_i \neq 1\}} p_i^{-1} \cdot (6\log r) \tau_i \sigma^2  g_i(z)^2 \le (6\log r) \rho^{-1} \sigma^2  f_{\G}(z) \le 6 \rho^{-1} \sigma^2 \log r
\end{align*}
where the last inequality used $z \in B$. Thus for $\zeta = 12\rho^{-1} \sigma^2 \log r$, Markov's inequality implies $\Pr_{y \sim \gamma_\sigma} \left(y \in B^0  \right)\geq 1/2$: thus $\gamma_{\sigma}(B^0) \ge  1/2$.
Now we use a standard trick to lower bound $\gamma_{\sigma}(B^x)$, by using that $B^0$ is symmetric. Notice that $d\gamma_{\sigma}(x+y) + d\gamma_{\sigma}(x-y) \ge 2e^{-\|x\|_2^2/(2\sigma^2)} d\gamma_{\sigma}(y)$. By symmetry of $B^0$,
\begin{align*} 
\gamma_{\sigma}(B^x) &= \frac{1}{2} \int_{y \in B^0} (d\gamma_\sigma (x+y) + d\gamma_\sigma (x-y)) \ge e^{-\|x\|_2^2/(2\sigma^2)} \int_{y \in B^0} d\gamma_\sigma (y) \\ &= e^{-\|x\|_2^2/(2\sigma^2)} \gamma_{\sigma}(B^0) \ge \frac{1}{2}e^{-\|x\|_2^2/(2\sigma^2)} \ge \frac{1}{2}e^{-1/(2\sigma^2)}. 
\end{align*}
Finally, we use this Gaussian measure estimate to  construct the desired sets $P_1, \dots, P_{S_{N-2}}$. We do so greedily: say that we have created $S$ sets $P_1, \dots, P_S$, and $x_{S+1} \in B_w$ is not in any of them. Then set $P_{S+1} \defeq x_{S+1} + 5B^0$. At the end of the process, $x_i - x_j \notin 5B^0$ for any $i \neq j$. We claim that $B^{x_i}$ are all disjoint. Indeed, if $u \in B^{x_i} \cap B^{x_j}$, then
\begin{align*} 
&\sum_{i\in[k]} \one_{\{p_i \neq 1\}} p_i^{-1} g_i((\mA^\top \mW \mA)^{-1/2}(x_i-x_j))^2 g_i(z)^2 \\ \le ~&2\sum_{i\in[k]} \one_{\{p_i \neq 1\}} p_i^{-1} (g_i((\mA^\top \mW \mA)^{-1/2}(x_i-u))^2 + g_i((\mA^\top \mW \mA)^{-1/2}(y_j-u))^2) g_i(z)^2 \le 4\zeta,
\end{align*} 
so $x_i-x_j \in 4B_0$, a contradiction. Because each $\gamma_{\sigma}(B^{x_i}) \ge \frac12 e^{-1/(2\sigma)^2}$, there are at most $2e^{1/(2\sigma^2)}$ sets $P_i$. Choose $\sigma^2 = 16 \cdot 2^{-N}$ and note $2e^{1/(2\sigma^2)} < 2^{2^{N-2}}$: the claim follows from $\zeta = 192 \cdot 2^{-N} \rho^{-1} \log r$ and setting $\Cltwo = 192$. 
\end{proof}
We are now ready to verify the growth condition.
\begin{lemma}
\label{lemma:cstar}
For $N \le \Cgf \log m$ and a universal constant $\Cstar$, the functions $F_N$ defined in \eqref{eq:deff} satisfy the growth condition described in \eqref{eq:growth} for $\lambda = \Clamb$ and $c^* = \Cstar \rho^{1/2}/\sqrt{\log m \log r}$.
\end{lemma}
\begin{proof}
Let $H_1, \dots, H_{2^{2^N}}$ be a collection of $(a,\lambda)$-separated sets, and let $v_i = \argmin\{\|x\|_{w+f} : x \in \hull(H_i)\}$. By \cref{lemma:growth1}, it suffices to show that $\Csc R^2 + \frac{1}{\log m} \ge c^* a 2^{N/2}$, where 
\[
R^2 \defeq \max_{j, j' \in [S_N]} \left( \|v_j - v_{j'}\|_{\mA^\top \mW \mA}^2 + \sum_{i\in[k]} (g_i(v_j) - g_i(v_{j'}))^2\right).
\]
We will now bound $a$ in terms of $R^2$. Fix $v = z = v_1 \in B$, and let $P_1,\dots,P_{2^{2^{N-2}}}$ be defined as in \cref{lemma:infcovering} and $Q_1,\dots,Q_{2^{2^{N-2}}}$ as in \cref{lemma:l2covering}. Because $S_{N-2}^2 < S_N$, the pigeonhole principle implies there are distinct $j, j' \in [S_N]$ such that $v_j, v_{j'}$ are in the same subsets $P_{\ell}$ and $Q_{\ell'}$. By \cref{lemma:well_separated}, we know that
\begin{align}
\label{eq:d_lb}
\wh{d}(v_j, v_{j'})^2 \ge a^2/(4 \Ct^4).
\end{align}
On the other hand, \cref{lemma:d_bound} implies 
\begin{align}
\label{eq:dbound_final}
\wh{d}(v_j, v_{j'})^2 &\leq 3 \wh{d}^\infty(v_j,v_{j'}) \sum_{i \in [k]} (g_i(v_j)-g_i(z))^2 + 3  \wh{d}^\infty(v_j,v_{j'}) \sum_{i \in [k]} (g_i(v_{j'})-g_i(z))^2 + 12 \wh{d}_z(v_j,v_{j'}) \nonumber \\
&\leq  6 R^2 \wh{d}^\infty(v_j,v_{j'}) + 12 \wh{d}_z(v_j,v_{j'})
\end{align}
where the last inequality used $\sum_{i \in [k]} (g_i(v_j)-g_i(z))^2 \leq R^2$ and $\sum_{i \in [k]} (g_i(v_{j'})-g_i(z))^2 \leq R^2$ by the definition of $R$ and $z=v_1$. By \cref{lemma:infcovering,lemma:l2covering}, we observe
\begin{align}
\label{eq:ball_bounds}
\wh{d}^\infty(v_j,v_{j'}) \leq \Cinfc 2^{-N} \rho^{-1} R^2 \log m \quad \text{and} \quad \wh{d}_z(v_j,v_{j'}) \leq \Cltwo 2^{-N} \rho^{-1} R^2 \log r.
\end{align}
we have
\begin{align*}
a^2/(4 \Ct^4) &\overset{(i)}{\leq} \wh{d}(v_j, v_{j'})^2 \\
&\overset{(ii)}{\leq}  6 R^2 \wh{d}^\infty(v_j,v_{j'}) + 12 \wh{d}_z(v_j,v_{j'}) \\
&\overset{(iii)}{\leq} 6 \Cinfc 2^{-N} \rho^{-1} R^4 \log m + 12 \Cltwo 2^{-N} \rho^{-1} R^2 \log r
\end{align*}
where $(i)$ used \eqref{eq:d_lb}, $(ii)$ used \eqref{eq:dbound_final}, and $(iii)$ used \eqref{eq:ball_bounds}. Now, we must have either of 
\[
\frac{a^2}{8 \Ct^4} \leq 6 \Cinfc 2^{-N} \rho^{-1} R^4 \log m  \quad \text{or} \quad \frac{a^2}{8 \Ct^4} \leq 12 \Cltwo 2^{-N} \rho^{-1} R^2 \log r  .
\]
In the former case, we have
\[
\frac{\Csc a 2^{N/2} \rho^{1/2}}{\sqrt{48 \Cinfc \log m} \Ct^2} \leq \Csc R^2.
\]
In the latter we get
\[ 
\Csc R^2 + \frac{1}{\log m} \ge \frac{\Csc a^2 \rho 2^N}{96 \Ct^4 \Cltwo \log r} + \frac{1}{\log m} \ge \frac{\sqrt{\Csc} }{\sqrt{24 \Cltwo} \Ct^2} \frac{a\rho^{1/2}2^{N/2}}{\sqrt{\log r \cdot \log m}}. 
\]
where we use the inequality $a+b \geq 2 \sqrt{ab}$. The desired relation between $R$ and $a$ (and thus the conclusion) follows with the choice
\[
\Cstar \geq \min \left\{ \frac{\Csc}{\Ct^{2} \sqrt{48 \Cinfc} }, \frac{\sqrt{\Csc} }{\Ct^2 \sqrt{24 \Cltwo} } \right\}.
 \]
\end{proof}
To complete the analysis of $\gamma$ in \eqref{eq:gamma}, we must describe how to handle the cases where $N \ge \Cgf \log m$. For this, we slightly modify the growth function. Define
\begin{align}
    \wh{F}_N(A) \defeq F_N(A) + c^* \sum_{M \ge \max\{\Cgf \log m, N\}} 2^{M/2+1} \cdot 2^{-2^{M-1}/m} \label{eq:newf}
\end{align}
for the value of $c^*$ in \cref{lemma:cstar}. We observe that $\wh{F}_N(A) \geq F_N(A)$ and $\wh{F}_{N+1}(A) \leq \wh{F}_N(A)$ for all $N$ and $A$.
\begin{lemma}
\label{lemma:alln}
The function $\wh{F}_N$ as defined in \eqref{eq:newf} satisfies the growth condition \eqref{eq:growth} for all $N \ge 0$ for the values of $\lambda$ and $c^*$ in \cref{lemma:cstar}.
\end{lemma}
\begin{proof}
Fix a collection of $(a,\lambda)$-separated sets $H_1, \dots, H_{2^{2^N}} \subseteq B$. Note that 
\begin{align*}
&\wh{F}_N\left(\bigcup_{i=1}^{S_N} H_i \right) - \min_{1 \le i \le S_N} \wh{F}_{N+1}(H_i) \ge F_N\left(\bigcup_{i=1}^{S_N} H_i \right) - \min_{1 \le i \le S_N} F_{N+1}(H_i)
\end{align*}
for all $N \ge 0$. Thus the cases $N \le \Cgf \log m$ follow from \cref{lemma:cstar}. For $N \ge \Cgf \log m$, we have
\begin{align*}
&\wh{F}_N\left(\bigcup_{i=1}^{S_N} H_i \right) - \min_{1 \le i \le S_N} \wh{F}_{N+1}(H_i) \ge c^* \cdot 2^{N/2+1} \cdot 2^{-2^{N-1}/m}.
\end{align*}
Our claim follows if $2^{N/2+1} \cdot 2^{-2^{N-1}/m} \geq a 2^{N/2}$. To establish this, we mirror \cref{lemma:largen} and argue that the disjointness of the $2^{2^N}$ $H_i$ implies that $a$ must be small. For any $j \in S_i$ and $x \in B$, observe
\[
\one_{\{p_i \neq 1\}} p_i^{-1} \l a_j , x \r^2 \leq \one_{\{p_i \neq 1\}} p_i^{-1} g_i(x)^2 \leq \one_{\{p_i \neq 1\}} p_i^{-1} \tau_i x^\top \mA^\top \mW \mA x \leq \rho^{-1} < 1
\]
by \cref{lemma:over,lemma:xawax}. Let $q_i = \one_{\{p_i \neq 1\}} p_i^{-1}$.
Observe $q_i^{1/2} | \l a_j , x \r | < 1$ for all $j \in S_i$ and $x \in B$. For each $x \in B$ and $\delta = 2^{-2^{N-1}/m}$ consider the vector $v^x \in \R^m$ defined as
\[ v^x_j \defeq \delta \lfloor q_i^{1/2} \l a_j , x \r / \delta \rfloor. \]
Note that $v^x$ can only be one of $\leq (1/\delta)^m = 2^{2^{N-1}}$ distinct vectors. If $x_s$ is the center of $H_s$ defined in \cref{def:separated}, the pigeonhole principle implies $v^{x_s} = v^{x_{s'}}$ for some $s \neq s'$: this implies 
\[ \left| q_i^{1/2} \l a_j , x_s - x_{s'} \r \right| \leq \delta  \]
for any $j$, and thus 
\[
\one_{\{p_i\neq1\}} p_i^{-1} g_i(x_s - x_{s'})^2 = \max_{j \in S_i} \one_{\{p_i\neq1\}} p_i^{-1} \l a_j , x_s - x_{s'} \r^2 \leq  \max_{j \in S_i} q_i \l a_j , x_s - x_{s'} \r^2 \leq \delta^2.
\]
for all $i \in [k]$. 
But now, well-separatedness (\cref{def:separated}) implies
\begin{align*}
    a^2 \le \wh{d}(x_s, x_{s'})^2 &= \sum_{i \in [k]} \one_{\{p_i \neq 1\}} p_i^{-1} (g_i(x_s) + g_i(x_{s'}))^2 g_i(x_s-x_{s'})^2  \\
    &\leq  2 \sum_{i \in [k]} \delta^2 (g_i(x_s)^2 + g_i(x_{s'})^2) \leq 4 \delta^2
\end{align*}
where the final inequality uses that $x_s, x_{s'} \in B$. The claim follows by the definition of $\delta$. 
\end{proof}
To finish, we simply need to apply \cref{thm:enough}. To do this, we must bound $\wh{F}_0(B)$ and $\mathrm{diam}_{\wh{d}}(B)$:

\begin{lemma}
\label{lem:diam_F_bound}
For $\Cgf \geq 9$ and $\wh{F}, \wh{d}$ defined above, we have $\wh{F}_0(B) \leq 1 + \Cgf + \frac{c^*}{100 (\Cgf - 9) m}$ and $\mathrm{diam}_{\wh{d}}(B) \leq 4 \rho^{-1/2}$. 
\end{lemma}
\begin{proof}
Note $\wh{F}_0(B) \leq 1 + \Cgf +  c^* \sum_{M \ge \Cgf \log m} 2^{M/2+1} \cdot 2^{-2^{M-1}/m}$. Further, for $M > 6 \log m$ we have $4m M < 2^{M}$ as $m \geq 2$ without loss of generality. Thus, we have 
\[
\wh{F}_0(B) \leq 1 + \Cgf +  2 c^* \sum_{M \ge \Cgf \log m} 2^{-2^{M}/(4m)} \leq 1 + \Cgf + \frac{c^*}{100 (\Cgf - 9) m}. 
\]
We now bound $\mathrm{diam}_{\wh{d}}(B)$. Let $x,y \in B$ be given. We observe
\begin{align*}
g_i(x-y)^2 &\overset{(i)}{\le} \tau_i \cdot (x-y)^\top \mA^\top \mW \mA (x-y) \leq 2 \tau_i x^\top \mA^\top \mW \mA x + 2 \tau_i y^\top \mA^\top \mW \mA y \\
&\overset{(ii)}{\le} 2 \tau_i f_{\G}(x) + 2 \tau_i f_{\G}(y) \overset{(iii)}{\le} 4 \tau_i
\end{align*}
where $(i)$ follows from \cref{lemma:over}, $(ii)$ follows from \cref{lemma:xawax}, and $(iii)$ follows from $x,y \in B$. Additionally, observe $\one_{\{p_i \neq 1\}} p_i^{-1} \tau_i \leq \rho^{-1}$ by definition of $p_i$. The result then follows as 
\begin{align*}
\wh{d}(x,y)^2 &= \sum_{i \in [k]} \one_{\{p_i \neq 1\}} p_i^{-1} (g_i(x) + g_i(y))^2 g_i(x-y)^2 \\
&\leq 4 \sum_{i \in [k]} \one_{\{p_i \neq 1\}} p_i^{-1} \tau_i (g_i(x) + g_i(y))^2 \\
&\leq 4 \rho^{-1}  \sum_{i \in [k]} (g_i(x) + g_i(y))^2 \\
&\leq 8 \rho^{-1} \sum_{i \in [k]} (g_i(x)^2 + g_i(y)^2) = 8 \rho^{-1} (f_{\G}(x)+ f_{\G}(y)) \leq 16 \rho^{-1}.
\end{align*}
\end{proof}
We are now ready to complete the proof of \cref{thm:chaining_bound}.
\begin{proof}[Proof of \cref{thm:chaining_bound}]
We again observe that we may assume $\eps > 1/m$ and $m \geq n$, as otherwise we may simply return $\G$ as our output sparsifier. By \cref{thm:chaining}, it suffices to show $\gamma \leq \frac{\eps}{\Cchain}$, where 
\begin{align} \gamma \defeq \inf_{\substack{T_s, T_{s+1}, \dots \\ T_N \subseteq B, |T_N| \le 2^{2^N} \text{ for all } N \ge s}} \sup_{x \in B} 2^{s/2} \cdot d(x, \vec{0}) + \sum_{N \ge s} 2^{N/2} d(x, T_N)
\end{align}
for any $s \geq \lceil \log_2 \log n \rceil$. We observe that the formula for $\gamma$ above depends on $d$ defined in \cref{sec:dudley}, \emph{not} the modified distance $\wh{d}$ we analyzed in this section. Our use of $\wh{d}$ is solely to employ the growth functional framework of \cite{book}.
We use \cref{lemma:dx0} to bound $d(x, \vec{0}) \leq \rho^{-1/2}$. Thus, for $s = \lceil \log_2 \log n \rceil$ the first term above is 
\begin{align}
\label{eq:term_1}
 \sup_{x \in B} 2^{s/2} \cdot d(x, \vec{0}) \leq 2 \sqrt{\log n} \cdot \rho^{-1/2} \leq \frac{2 \eps \sqrt{\log m}}{\sqrt{\Cmain \log m \log r}} \leq \frac{2 \eps}{\sqrt{\Cmain \log 2}}
\end{align}
as $r \geq 2$ (as we assumed in \cref{sec:intromatrix}). For the second term, we have
\begin{align}
\inf_{\substack{T_s, T_{s+1}, \dots \\ T_N \subseteq B, |T_N| \le 2^{2^N} \text{ for } N \ge s}} \sup_{x \in B} \sum_{N \ge s} 2^{N/2} d(x, T_N) &\overset{(i)}{\le} \inf_{\substack{T_0, T_{1}, \dots \\ T_N \subseteq B, |T_N| \le 2^{2^N} \text{ for } N \ge 0}} \sup_{x \in B} \sum_{N \ge 0} 2^{N/2} \wh{d}(x, T_N) \nonumber \\
&\overset{(ii)}{\leq} \Ctal \left( \frac{\lambda \wh{F}_0(B)}{c^*} + \lambda \mathrm{diam}_{\wh{d}}(B) \right).\label{eq:term_2}
\end{align}
$(i)$ uses $d(x,y) \leq \wh{d}(x,y)$ for any $x,y \in B$, and $(ii)$ uses \cref{thm:enough} combined with \cref{lemma:alln}, where $\lambda = \Clamb$ and $c^* = \Cstar \rho^{1/2} / \sqrt{\log m \log r}$. We now use \cref{lem:diam_F_bound} to obtain
\begin{align}
\Ctal \left( \frac{\lambda \wh{F}_0(B)}{c^*} + \lambda \mathrm{diam}_{\wh{d}}(B) \right) &\leq \Ctal \Clamb \left( \frac{1 + \Cgf}{c^*} + \frac{1}{100 (\Cgf - 9) m} + 4 \rho^{-1/2} \right)  \nonumber \\
&\leq \Ctal \Clamb \left( \frac{(1 + \Cgf) \eps}{\Cstar \sqrt{\Cmain}} + \frac{\eps}{100 (\Cgf - 9) } + \frac{4 \eps}{\sqrt{\Cmain} \sqrt{\log m \log r}} \right). \label{eq:term_22}
\end{align}
Choose $\Cgf = 9 + \Cchain \Ctal \Clamb$: combining \cref{eq:term_1}, \cref{eq:term_2}, and  \cref{eq:term_22} gives 
\[
\gamma \leq \frac{2 \eps}{\sqrt{\Cmain \log 2}} + \Ctal \Clamb \left( \frac{(10 + \Cchain \Ctal \Clamb) \eps}{\Cstar \sqrt{\Cmain}} + \frac{\eps}{100 \Cchain \Ctal \Clamb } + \frac{4 \eps}{\sqrt{\Cmain} \log 2 }\right)
\]
where we again use $m, r \geq 2$. We choose 
\[
\Cmain = 900 \Cchain^2 + 900 \Cchain^2 \Ctal^2 \Clamb^2 + 9  \Cstar^{-2} (10 + \Cchain \Ctal \Clamb)^4.
\]
As $\sqrt{\Cmain} > 30 \Cchain $ and $\sqrt{\Cmain} > 30 \Cchain \Ctal \Clamb$ we have 
\begin{align*}
\gamma &\leq \frac{2 \eps}{30 \Cchain \log 2} + \Ctal \Clamb \left( \frac{(10 + \Cchain \Ctal \Clamb) \eps}{3 (10 + \Cchain \Ctal \Clamb)^2} + \frac{\eps}{100 \Cchain \Ctal \Clamb } + \frac{4 \eps}{30 \Cchain \Ctal \Clamb \log 2 }\right) \\
&\leq \frac{9 \eps}{30 \Cchain} + \frac{\eps}{100 \Cchain} + \frac{ \Ctal \Clamb \eps}{3 (10 + \Cchain \Ctal \Clamb)} \leq \frac{\eps}{\Cchain}.
\end{align*}
The result follows. 
\end{proof}

\cref{thm:matrix} follows from \cref{thm:chaining_bound}, and runtime follows from \cref{thm:overestimation_general_hypergraph}. Now, \cref{thm:hypergraph} follows with the weaker sparsity bound $O(n \eps^{-2} \log m \log r)$ from \cref{thm:chaining_bound} and our algorithm for computing leverage score overestimates in \cref{thm:overestimation_graph_hypergraph}.

Finally, we discuss how to improve the $\log m$ to $\log n$ in \cref{thm:chaining_bound} in the case of graphical hypergraphs. Let $\G$ be a graphical hypergraph with hyperedges $\S$ and weights $v$, and let $\G = (\S, \mA)$ be its representation as an unit hypergraph, where the initial weights have been incorporated into $\mA$. We will use that the number of distinct rows in $\mA$ up to scaling is at most $n(n-1)/2$ and modify \cref{lemma:infcovering,lemma:alln}. Let us consider the setting of \cref{lemma:infcovering} for concreteness. For $u_1, u_2 \in [n]$, let $E_{u_1, u_2} \subseteq [k]$ be the set of hyperedges containing both $u_1$ and $u_2$. We define
\[
i(u_1, u_2) = \argmax_{i \in E_{u_1, u_2}}  \one_{\{p_i \neq 1\}} p_i^{-1} v_i \quad \text{and} \quad q(u_1, u_2) = \max_{i \in E_{u_1, u_2}}  \one_{\{p_i \neq 1\}} p_i^{-1} v_i.
\]
Let $j(u_1,u_2) \in S_{i(u_1,u_2)}$ be the row of $\mA$ which is a multiple of $\indicVec{u_1} - \indicVec{u_2}$. Let $\rows$ be the set of all $j(u_1, u_2)$, and note that $|\rows| \le n(n-1)/2$.
Now, in the proof of \cref{lemma:infcovering}, repeat the argument, except initially restrict to only considering the rows in $\rows$. The guarantee in \eqref{eq:akguarantee} shows that $\l a_{j(u_1,u_2)}, x-y\r^2 \le \tau_{i(u_1,u_2)}\eta^2$, so using that $a_{j(u_1,u_2)} = v_{i(u_1,u_2)}^{1/2}(\indicVec{u_1} - \indicVec{u_2})$ gives
\[ \one_{\{p_{i(u_1,u_2)}\neq 1\}}p_{i(u_1,u_2)}^{-1} v_{i(u_1,u_2)} \l \indicVec{u_1} - \indicVec{u_2}, x-y \r^2 \le  \one_{\{p_{i(u_1,u_2)}\neq 1\}}p_{i(u_1,u_2)}^{-1} \tau_{i(u_1,u_2)}\eta^2 \le \rho^{-1}\eta^2, \]
for all $x, y \in P_{\ell}$ for some $\ell$, and $\rho = \eps^{-2}\log n \log r$ now. Let $j \in S_i$ be such that the corresponding edge is also $(u_1,u_2)$. By maximality of $i(u_1, u_2)$ and $j(u_1, u_2)$ we deduce that \[ \one_{\{p_i \neq 1\}} p_i^{-1} \l a_j, x-y \r^2 = \one_{\{p_i \neq 1\}} p_i^{-1}v_i \l \indicVec{u_1} - \indicVec{u_2}, x-y \r^2 \le \rho^{-1} \eta^2 \] which gives our desired bound on $\wh{d}^\infty$. The analogous definitions and argument similarly apply to \cref{lemma:alln}. These bounds can be plugged in to verify the growth condition for a modified $F_N(\cdot)$ where the $\log m$ is replaced with $\log n$ in \eqref{eq:deff}, as is done in \cref{lemma:alln}.

\section*{Acknowledgments}

We thank James Lee for coordinating submissions.

Yang P. Liu is supported by the Google PhD Fellowship Program. Aaron Sidford is supported by a Microsoft Research Faculty Fellowship, NSF CAREER Award CCF-1844855, NSF Grant CCF-1955039, a PayPal research award, and a Sloan Research Fellowship.

{\small
\bibliographystyle{alpha}
\bibliography{refs}}

\end{document}